\documentclass[12pt]{article}
\usepackage{graphicx} 
\usepackage{graphicx} 
\usepackage{bm}
\usepackage{tabularx}
\usepackage{amsthm, amsmath}
\usepackage{amssymb}
\usepackage{hyperref}
\usepackage[ruled, vlined, linesnumbered]{algorithm2e}
\usepackage{booktabs}
\usepackage{enumitem}
\usepackage[table,xcdraw]{xcolor} 
\usepackage{multirow}
\usepackage{tikz} 
\usetikzlibrary{patterns} 

\theoremstyle{plain}
\newtheorem{theorem}{Theorem}[section]
\newtheorem{lemma}[theorem]{Lemma}

\theoremstyle{definition}
\newtheorem{assumption}{Assumption}

\theoremstyle{definition}
\newtheorem{definition}[theorem]{Definition}

\theoremstyle{remark}

\usepackage{xcolor}
\usepackage{tcolorbox}
\usepackage{graphicx}
\usepackage{subcaption} 
\usepackage{pgfplots}
\usepgfplotslibrary{groupplots, fillbetween}
\pgfplotsset{compat=1.18}

\SetKwInput{KwInput}{Input}
\SetKwInput{KwOutput}{Output}
\SetKwComment{Comment}{// }{}
\AtBeginDocument{%
  }

\title{Undetectable Backdoors in Model Parameters: Hiding Sparse Secrets in High Dimensions}
\author{
  Sarthak Choudhary$^{1}$ \quad
  Atharv Singh Patlan$^{2}$ \quad
  Nils Palumbo$^{1}$ \\
  Ashish Hooda$^{1}$ \quad
  Kassem Fawaz$^{1}$ \quad
  Somesh Jha$^{1}$ \\[0.5em]
  \normalsize $^{1}$University of Wisconsin--Madison \quad
  $^{2}$Princeton University
}
\date{}

\begin{document}

\maketitle

\begin{abstract}
    \noindent The widespread adoption of pre-trained models distributed through public repositories such as Hugging Face has introduced a supply-chain attack surface in which downstream consumers must rely on classifiers from untrusted third parties. Such a classifier may behave correctly on clean inputs but route trigger-embedded inputs to an adversary-chosen target class. Parameter-level detection is the primary line of defense against such attacks, yet existing detectors and attacks have co-evolved empirically, with no attack to date ruling out detection by any efficient algorithm. The only prior construction with a formal undetectability guarantee is restricted to single-layer networks with weights drawn from a random distribution, leaving open whether provable undetectability is achievable for the pre-trained multi-layer classifiers used in practice.

    We present \textbf{Sparse Backdoor}, a supply-chain attack that plants a \emph{provably undetectable} backdoor in pre-trained image classifiers, including convolutional networks and Vision Transformers. The attack injects a structured sparse perturbation along a randomly chosen direction into a small subset of columns at each fully connected layer, propagating a trigger signal to an adversary-chosen target class, and masks the perturbation with an independent isotropic Gaussian dither. The dither serves a single technical purpose: it induces a clean reference distribution anchored at the pre-trained weights, against which undetectability can be formalized. Under a mild margin condition on the pre-trained classifier, we show that the dithered reference is functionally equivalent to the original classifier. We prove that distinguishing the backdoor-injected model from this reference is \emph{at least as hard as Sparse PCA detection}, which is computationally infeasible under standard hardness assumptions. The guarantee holds against any probabilistic polynomial-time distinguisher with white-box access to the parameters.

    Across nine architecture--dataset configurations on CIFAR-10, SVHN, and GTSRB, Sparse Backdoor exceeds $\bm{93\%}$ attack success on CIFAR-10 while preserving clean accuracy within $\bm{1.5}$--$\bm{8}$ points of the baseline, and evades Neural Cleanse, FeatureRE, and UNICORN at a mean distinguishing advantage of $\bm{0.12}$, close to random guessing. These results show that detection of parameter-level backdoors is fundamentally limited, and motivate a shift toward mitigation strategies that neutralize backdoors without identifying them. Our implementation is publicly available at \url{https://github.com/sarthak-choudhary/Sparse_Backdoors}.
\end{abstract}

\section{Introduction}

Training high-performance machine learning models is computationally expensive and often requires considerable domain expertise. To lower this barrier, platforms such as Hugging Face~\cite{huggingface}, TensorFlow Model Garden~\cite{tfmodelgarden}, and ModelZoo~\cite{modelzoo} allow model providers to share pre-trained models with downstream users who lack the resources to train from scratch. On Hugging Face alone, any third party can publish a model without restriction, making it immediately available for download by the entire community. This reliance on third-party models exposes a practical threat: \emph{backdoor attacks}~\cite{gu2017badnets, chen2017targeted, turner2018clean, yao2019latent, nguyen2020input, nguyen2021wanet, liu2018trojaning, bagdasaryan2021blind, cao2024data, xu2025towards, rakin2020tbt, hong2022handcrafted}. A malicious provider can supply a model that performs correctly on standard inputs yet exhibits attacker-chosen behavior when presented with a specific \emph{trigger}. Such compromised models create serious vulnerabilities in any downstream application that deploys them.

In this supply-chain setting, the model consumer has no visibility into how a model was trained and must judge its trustworthiness solely from the delivered parameters. The attacker, by contrast, has full control and can inject backdoors by poisoning training data~\cite{chen2017targeted, gu2017badnets, turner2018clean, yao2019latent}, manipulating the training procedure~\cite{bagdasaryan2021blind, liu2018trojaning}, or directly modifying the model's weights~\cite{cao2024data, hong2022handcrafted, rakin2020tbt}. This asymmetry makes parameter-level detection the primary line of defense, allowing a platform to reject compromised models before they are made available. Existing defenses attempt this detection by analyzing the delivered model, whether by inspecting its parameters directly or by probing its behavior on chosen inputs to surface anomalies in activations or outputs, often by reconstructing potential triggers~\cite{wang2019neural, wang2022featurere, wang2023unicorn}. In practice, however, defenses and attacks have co-evolved in an empirical cat-and-mouse cycle: each new defense exploits artifacts left by current attacks, and subsequent attacks learn to suppress those artifacts. Even attacks that are explicitly designed for stealth only provide guarantees against specific named defenses~\cite{cao2024data} or rely on heuristics without any formal guarantee at all~\cite{xu2025towards}. \textit{No existing attack rules out detection by an arbitrary probabilistic polynomial-time (PPT) distinguisher}, leaving the cycle unresolved and the possibility of stronger defenses open.

Goldwasser et al.~\cite{goldwasser2022planting} were the first to place backdoor undetectability on cryptographic footing, proving that no efficient algorithm can distinguish a backdoored model from a clean one. Their white-box construction achieves this indistinguishability for models built on Random Fourier Features and single-layer Random ReLU networks, architectures chosen specifically to enable reduction to cryptographic hardness assumptions. While a significant theoretical contribution, these architectures bear little resemblance to the deep networks deployed in practice, and the construction has not been adopted by the empirical backdoor literature. Subsequent attacks continue to evaluate stealth through ad hoc metrics~\cite{xu2025towards} or against specific named defenses~\cite{cao2024data, shokri2020bypassing}. A gap remains between theoretical and empirical contributions on backdoor attacks: \textit{can we achieve provable undetectability for the standard architectures actually deployed in practice?}

To close this gap, we present the \textbf{Sparse Backdoor attack}, which plants provably undetectable backdoors in standard architectures, including convolutional networks and Vision Transformers. Given a pre-trained classifier from a benign training pipeline, the attack optimizes a small input-space trigger and injects structured, sparse perturbations into the fully connected layers. This attack propagates a backdoor signal layer by layer to the target behavior. To mask these perturbations, each modified weight column also receives an independent Gaussian dither. Unlike prior undetectability results for random-weight architectures~\cite{goldwasser2022planting}, where the clean parameter distribution is fixed by construction, pretrained supply-chain models lack a canonical distribution over clean weights. We therefore define undetectability with respect to \emph{functional cleanliness} rather than membership in a training-induced distribution: a sound detector should not reject any classifier that implements the intended task and contains no attacker-planted trigger. We instantiate this by applying a calibrated Gaussian dither to the pretrained weights to obtain a clean reference distribution. Under a mild margin condition, every model in this reference class computes the same function as the pretrained classifier and is therefore clean. Consequently, \emph{any sound backdoor detector must distinguish the backdoor-injected model from this reference class}. Since the only difference between the backdoored and reference models is a sparse component, Sparse PCA detection~\cite{berthet2013complexity, brennan2019optimal} reduces to distinguishing them, making backdoor detection computationally infeasible under standard cryptographic assumptions~\cite{barak2019nearly}. To our knowledge, it is the \textit{first backdoor attack on practical architectures with a formal undetectability guarantee against all efficient distinguishers}.

We evaluate the Sparse Backdoor attack across three architectures (ConvNet, ResNet-18, and Vision Transformer) and three datasets (CIFAR-10~\cite{krizhevsky2009learning}, SVHN~\cite{netzer2011reading}, and GTSRB~\cite{stallkamp2012man}). We find that it reliably achieves high attack success while evading three representative detection methods~\cite{wang2019neural, wang2022featurere, wang2023unicorn} spanning parameter, feature, and input-space analysis. Fine-tuning, a common post-hoc mitigation~\cite{liu2018finepruning}, reduces attack success inconsistently across configurations and is unreliable as a standalone defense.

\noindent In summary, our contributions are as follows:
\begin{itemize}[leftmargin=*, nosep]
    \item \textbf{A backdoor attack on practical architectures.} We propose the Sparse Backdoor attack, which plants backdoors in standard architectures by injecting structured sparse perturbations into fully connected layers. The attack exceeds $\bm{93\%}$ attack success rate on CIFAR-10 across all settings, reaching $\bm{99.5\%}$ on ConvNet, $\bm{93.5\%}$ on ResNet-18, and $\bm{99.6\%}$ on ViT-Small, while maintaining accuracy within $\bm{1.5}$ to $\bm{8.5}$ percentage points of the clean baseline.

    \item \textbf{A formal white-box undetectability guarantee.} We prove the resulting backdoored model is computationally indistinguishable from a clean classifier under the Sparse PCA hardness assumption. The proof introduces a margin-based functional equivalence argument showing that a Gaussian-dithered reference model computes the same function as the original baseline, and shows a reduction from Sparse PCA detection to backdoor detection.

    \item \textbf{Comprehensive empirical validation.} We empirically validate the attack across nine architecture-dataset configurations and three detection mechanisms (Neural Cleanse \cite{wang2019neural}, FeatureRE \cite{wang2022featurere}, and UNICORN~\cite{wang2023unicorn}). The overall mean distinguishing advantage is $\bm{0.12}$, close to the random-guessing baseline of $0.0$, confirming that existing detectors cannot separate the backdoored model from a clean one. We further verify the assumptions underlying our theoretical guarantee across all configurations. 
\end{itemize}
\smallskip
\noindent Our results demonstrate that even in the strongest possible setting for the defender, that is, white-box access to all model parameters, detecting the backdoor is computationally infeasible under standard hardness assumptions. This suggests that detection-based defenses are fundamentally limited against cryptographically grounded attacks and that the community should invest in mitigation strategies that can neutralize backdoors without first detecting them, as explored in recent theoretical work~\cite{goldwasser2025oblivious}.

\section{Related Work}
\subsection{Backdoor Attacks}
 Prior work can be grouped by the adversary’s level of control. No attack in either group provides a formal undetectability guarantee against arbitrary PPT distinguishers; stealth is argued through ad hoc metrics or against specific named defenses.

\smallskip 
\noindent \textbf{Data poisoning attacks.} In poisoning attacks, the adversary modifies part of the training data so that the model learns to associate a trigger with a target class~\cite{gu2017badnets, chen2017targeted, turner2018clean, yao2019latent}. Early methods used fixed visible patterns, such as patches in BadNets~\cite{gu2017badnets} and blended overlays in Blend~\cite{chen2017targeted}. Since such triggers are spotted by inspection or automated screening~\cite{chen2018detecting, wang2019neural}, later work pursued less perceptible triggers, including input-dependent triggers in IAD~\cite{nguyen2020input} and geometric warping in WaNet~\cite{nguyen2021wanet}. Despite these refinements, poisoning attacks still influence the model only indirectly through data, and recent work shows that they can leave separable signatures in learned representations~\cite{xu2025towards, wang2022rethinking}.

\smallskip
\noindent \textbf{Supply-chain attacks.} In the supply-chain setting, the adversary can modify the weights or training procedure directly~\cite{liu2018trojaning, bagdasaryan2021blind, cao2024data, xu2025towards}.  Blind backdoor attacks~\cite{bagdasaryan2021blind} add a malicious objective to the training loss, while weight perturbation attacks such as TBT~\cite{rakin2020tbt}, handcrafted backdoors~\cite{hong2022handcrafted}, and DFBA~\cite{cao2024data} directly manipulate model parameters. This grants fine-grained access to the model's internal state, unlike poisoning attacks, but their guarantees are typically either heuristic or tied to specific named defenses (like ~\cite{cao2024data, wang2019neuralcleanse, xu2021detecting}) rather than general efficient detectors. Our work instead studies undetectability under a computational hardness assumption.

\subsection{Backdoor Detection}
Detection-based defenses aim to determine whether a given model has been backdoored, typically by reverse-engineering a candidate trigger and testing whether it is anomalously effective. We focus on deployment-phase defenses, which operate on trained models, as these are the most relevant to our supply-chain threat model.

Neural Cleanse~\cite{wang2019neuralcleanse} optimizes a per-class perturbation that induces targeted misclassification and flags the model as backdoored if one class admits a trigger with anomalously small $L_1$ norm. It is effective against patch-based triggers but can miss triggers distributed across the input (e.g., blending or warping)~\cite{nguyen2021wanet, chen2017targeted}. FeatureRE addresses this by performing trigger inversion in feature space rather than input space, detecting backdoors through separability in learned representations. UNICORN~\cite{wang2023unicorn} unifies both perspectives through a generative trigger inversion framework that operates across input and feature spaces simultaneously.

Other deployment-phase strategies include fine-tuning on clean data to suppress backdoor behavior~\cite{zeng2021ibau}, and pruning methods that remove neurons exhibiting anomalous activation patterns~\cite{liu2018finepruning, wu2021anp}. TABOR~\cite{tabor2019tabor} augments trigger inversion with additional regularizers to handle triggers of varying shape and location, while BTI-DBF~\cite{xu2024btidbf} and BAN~\cite{xu2024ban} improve the efficiency and sensitivity of feature-space inversion. A common theme across all these defenses is that they are designed to exploit structural artifacts left by the attack. Our work shows that when the attack's perturbation structure is grounded in a computational hardness assumption, these artifacts become provably undetectable by any efficient algorithm.

\subsection{Provably Undetectable Backdoors}
Goldwasser et al.~\cite{goldwasser2022planting} initiated the study of backdoor attacks with formal undetectability guarantees. They presented two constructions: one based on digital signatures that achieves black-box undetectability, and one based on the hardness of Sparse PCA that achieves white-box undetectability in random Fourier feature networks and random ReLU networks. In the white-box construction, even with full access to the model parameters, no PPT distinguisher can separate the backdoored model from a clean one.

Their framework, however, is limited in two key respects. First, it applies only to single-layer networks whose weights are drawn from a known random distribution, which provides a natural null hypothesis for the detection problem. Pre-trained models have fixed weights whose distribution is not known in closed form, so this assumption does not hold. Second, the single-layer setting avoids the challenges of signal propagation through cross-layer dependencies that arise in multi-layer networks. These structural limitations leave the construction theoretical, with no empirical evaluation on any real architecture.

Subsequent works extend the framework along different axes. Kalavasis et al.~\cite{kalavasis2024injecting} use indistinguishability obfuscation (iO) to plant undetectable backdoors in arbitrary neural networks and language models, but iO remains impractical~\cite{jain2026indistinguishability, li2025succinct}. Ngo et al.~\cite{ngo2025cryptographic} extend the signature-based (black-box) construction to image classification, connecting cryptographic backdoors to adversarial robustness, watermarking, and IP protection. Neither addresses Sparse PCA-based white-box undetectability for multi-layer pre-trained classifiers.

\section{Preliminaries}
\label{sec:prelim}

We review the model setup and the hardness assumption underlying our undetectability guarantees; Table~\ref{tab:notation} summarizes notation.

\begin{table}[ht]
\centering
\renewcommand{\arraystretch}{1.0}
\setlength{\tabcolsep}{3pt}
\caption{Summary of Notation}
\label{tab:notation}
\resizebox{0.88\columnwidth}{!}{%
\begin{tabular}{@{}cl@{}}
\toprule
\textbf{Symbol} & \textbf{Description} \\
\midrule
\multicolumn{2}{l}{\textit{Model Architecture \& Data}} \\
\addlinespace[1pt]
$f, \tilde{f}$            & Clean classifier and backdoor-injected classifier. \\
$f_{enc}$                 & Frozen feature encoder. \\
$W_i, w_i^{(j)}$          & Weight matrix of $i$-th FC layer and its $j$-th column. \\
$x_i, d_i$                & Layer $i$ input embedding and input dimension. \\
$\mathcal{D}, X$          & Clean data distribution and a set of sample images. \\
\midrule
\multicolumn{2}{l}{\textit{Backdoor Attack Parameters}} \\
\addlinespace[1pt]
$\Delta^*$                & Optimized trigger. \\
$\mathcal{A}_{\Delta^*}$  & Poisoning function: $\mathcal{A}_{\Delta^*}(x) = x + \Delta^*$. \\
$\mathcal{I}_i$           & Set of column indices perturbed in layer $i$. \\
$\alpha$                  & Sparsity exponent $\alpha \in (0, 1/2]$. \\
$k_i$                     & Sparsity level at layer $i$, $k_i = \lfloor d_i^\alpha \rfloor$. \\
$s_i$                     & $k_i$-sparse unit backdoor signal at layer $i$. \\
$y_t$                     & Target classification label. \\
$\eta_i^{(j)}, \tau_i^2$  & Gaussian dither noise and its variance. \\
$\xi_i^{(j)}, \sigma_i^2$ & Backdoor spike coefficient and its variance. \\
\bottomrule
\end{tabular}%
}
\end{table}

\subsection{Model Architecture}
\label{sec:model_arch}

We focus on pretrained image classifiers whose prediction head consists of L fully connected (FC) layers with ReLU activations between hidden layers. This structure is shared by a wide range of modern architectures: in convolutional networks (e.g., ResNet), the encoder consists of convolutional and pooling layers; in vision transformers (e.g., ViT), it consists of patch embedding and transformer encoder blocks. In both cases, the encoder maps an input image to a fixed-dimensional feature vector, which is then processed by the FC layers to produce class logits. 

Formally, let $\mathcal{D}$ be a distribution over $\mathcal{X} \times \mathcal{Y}$, where $\mathcal{X}$ is the input space and $\mathcal{Y} = \{1, \ldots, C\}$ is the label space ($C$ being the number of classes). Let $f: \mathcal{X} \to \mathcal{Y}$ be a classifier trained on $\mathcal{D}$. We decompose $f$ into its feature encoder and the weight matrices of its $L$ fully connected layers, $f = \{f_{enc}, W_1, W_2, \dots, W_L\}$, where $f_{enc}$ denotes the feature encoder (frozen during the attack) and $W_i \in \mathbb{R}^{d_i \times d_{i+1}}$ is the weight matrix of the $i$-th FC layer. For an input image $x$, the network produces a logit vector 
\[
    g(x) \;=\;
    W_L^\top\!\left(
    \operatorname{ReLU}(W_{L-1}^\top(\cdots
    \operatorname{ReLU}(W_1^\top f_{enc}(x)))
    )\right) \;\in\; \mathbb{R}^C,
\]
 and the inference operation returns the predicted label $f(x) \;=\; \operatorname*{argmax}_{y \in [C]} g(x)_y$. We omit bias terms\footnote{A bias term is an additive vector $b_i \in \mathbb{R}^{d_{i+1}}$ applied after the linear map at layer $i$, giving the pre-activation $W_i^\top x_i + b_i$.} for brevity, and we omit the Softmax, since $\operatorname{argmax}$ is invariant to it.
 \smallskip

\noindent For the $i$-th FC layer with weight matrix $W_i = [w_i^{(1)}, \dots, w_i^{(d_{i+1})}] \in \mathbb{R}^{d_i \times d_{i+1}}$, and input embedding $x_i \in \mathbb{R}^{d_i}$, the output $x_{i+1} \in \mathbb{R}^{d_{i+1}}$ is 
\[
    x_{i+1} = \operatorname{ReLU}(W_i^\top x_i)
    = \left[\operatorname{ReLU}(\langle w_i^{(1)},
    x_i \rangle), \dots,
    \operatorname{ReLU}(\langle w_i^{(d_{i+1})},
    x_i \rangle)\right]^\top
\]
A key property that we exploit is that each column $w_i^{(j)}$ of $W_i$ exclusively controls the $j$-th output neuron of $x_{i+1}$. This column-wise independence allows the attacker to perturb individual neurons without affecting others, and is what connects the weight-space perturbation to the Sparse PCA detection problem introduced next.

\subsection{Sparse PCA}

We adopt Sparse PCA as the computational hardness assumption underlying our undetectability guarantees. Informally, the Sparse PCA detection problem asks whether it is computationally feasible to distinguish samples whose covariance contains a variance spike along an unknown sparse direction.

\begin{definition}[Sparse PCA Detection Problem~\cite{brennan2018reducibility, brennan2019optimal, goldwasser2022planting}] Let $k, d \in \mathbb{N}$ with $k \leq d$ and let $\theta > 0$. Let $Y = \{y_j\}_{j=1}^{t} \subset \mathbb{R}^d$ be a set of $t$ i.i.d.\ samples drawn from an unknown distribution, and let $v \in \mathbb{R}^d$ be an unknown $k$-sparse vector satisfying $\| v\|_0 = k$ and $\| v\|_2 = 1$. The goal is to distinguish between:

\[
\mathcal{H}_{\mathrm{null}}: y_j \stackrel{\text{i.i.d.}}{\sim} \mathcal{N}(0, I_d)
\quad \text{vs.} \quad
\mathcal{H}_{\mathrm{alt}}: y_j \stackrel{\text{i.i.d.}}{\sim} \mathcal{N}(0, I_d + \theta v v^\top).
\]
\end{definition}

\begin{assumption}[Sparse PCA Detection Hardness~\cite{brennan2018reducibility, brennan2019optimal, cs354Notes}]
\label{assm:hardness}
Let $k = \lfloor d^\alpha \rfloor$ for some $0 < \alpha \leq \tfrac{1}{2}$ and let $\theta =  o\!\left(k \sqrt{\frac{\log d}{t}}\right)$. Under these parameters, the Sparse PCA detection problem is computationally hard: for any probabilistic polynomial-time (PPT) detection algorithm $\mathcal{G} : \mathbb{R}^{t \times d} \to \{0,1\}$, there exists  $\varepsilon(d) = o(1)$ such that 
\[
\left|
\Pr_{y_j \stackrel{\text{i.i.d.}}{\sim} \mathcal{N}(0, I_d)}\!\big[\mathcal{G}(Y)=1\big]
-
\Pr_{y_j \stackrel{\text{i.i.d.}}{\sim} \mathcal{N}(0, I_d + \theta v v^\top)}\!\big[\mathcal{G}(Y)=1\big]
\right|
\le \varepsilon(d),
\]
where $v$ is an unknown $k$-sparse unit vector. In particular, no PPT algorithm can distinguish the spiked covariance distribution from the standard Gaussian with a constant advantage when the spike is hidden along an unknown sparse direction.
\end{assumption}

\noindent The computational hardness of the Sparse PCA detection problem is supported by its reduction from the \textit{Planted Clique} conjecture~\cite{berthet2013complexity, berthet2013optimal, wang2016statistical, brennan2018reducibility, gao2017sparse}. The planted clique problem asks to distinguish Erd\H{o}s--R\'{e}nyi random graphs on $n$ nodes from random graphs containing a planted $k$-clique. This problem is widely believed to be intractable for PPT algorithms when $k < \sqrt{n}$~\cite{barak2019nearly, feldman2017statistical}. Our attack operates precisely in this hard regime, with the backdoor signal-to-noise ratio calibrated to fall below the computational detection threshold.

\section{Attack Formalization}

A backdoor attack maps inputs with an adversary-chosen \textit{trigger} to a \textit{target class} while preserving clean behavior. We formalize the objective, undetectability via a security game, and threat model.

\subsection{Backdoor Objective}
\label{sec:formal_attack}
A classifier $f$ is \emph{well-trained} on $\mathcal{D}$ if its misclassification probability over samples drawn from $\mathcal{D}$ is bounded by a small constant $\epsilon > 0 $:

\[
    \Pr_{(x, y) \sim \mathcal{D}}
    [f(x) \neq y] < \epsilon,
\]

\noindent A classifier $\tilde{f}$ is \emph{backdoor-injected} with respect to a distribution $\mathcal{D}$, a target class $y_t \in \mathcal{Y}$, and a poisoning function $\mathcal{A}_\Delta: \mathcal{X} \to \mathcal{X}$ that maps each clean input $x$ to its triggered counterpart $\mathcal{A}_{\Delta}(x) = x + \Delta$ for a trigger perturbation $\Delta \in [-\delta, \delta]^d$ with a small constant $\delta > 0 $, if it satisfies the following two conditions:
\begin{itemize}[leftmargin=*, nosep]
    \item \textbf{Attack Success:} The model misclassifies poisoned inputs as the target class with high probability over the input distribution:
    \begin{equation}
        \Pr_{(x,y) \sim \mathcal{D}} \bigl[\tilde{f} \left(\mathcal{A}_{\Delta}(x)\right) = y_t\bigr] \ge 1  -\epsilon_{\mathrm{atk}} \quad \text{where } y_t \neq y
        \label{eq:ASR}
    \end{equation}
    for a small constant $\epsilon_{\mathrm{atk}} > 0$.

\item \textbf{Clean Accuracy Preservation:} The model remains accurate on clean inputs drawn from the same distribution:
    \begin{equation}
        \Pr_{(x, y) \sim \mathcal{D}}
        [\tilde{f}(x) \neq y] \;\leq\; \epsilon + c
        \label{eq:utility}
    \end{equation}
    for a small constant $c \geq 0$ bounding the drop in clean accuracy, calibrated to remain unnoticeable for the downstream task.
\end{itemize}

\noindent Beyond these minimal requirements, an attack may also aim to evade detection defenses that inspect model parameters. We formalize this next by requiring the parameters of $\tilde{f}$ to be computationally indistinguishable from those of a clean model.

\smallskip
\noindent \textbf{Functional cleanliness.}
A classifier $f$ is \emph{clean} for $\mathcal{D}$ if it implements the intended task behavior on clean inputs and is not susceptible to an attacker-planted trigger that induces targeted behavior. The condition is purely behavioral: if $f$ is clean, any classifier $f'$ that produces the same predicted label as $f$ on every input is also clean. We call a probability distribution a \emph{clean-reference distribution} if its support consists of functionally clean classifiers.


\subsection{Undetectability}
\label{section:undetectability}
Following Goldwasser et al.~\cite{goldwasser2022planting}, a backdoor attack is \emph{undetectable} if the resulting classifier $\tilde f$ is computationally indistinguishable, under white-box access, from a classifier sampled from a clean-reference distribution. We formalize this through an indistinguishability game (Figure~\ref{fig:security_game}) between a challenger $\mathcal{C}$, who runs the experiment, and a PPT distinguisher $\mathcal{G}$, who attempts to distinguish a backdoor-injected classifier from a classifier sampled from this distribution. We take the security parameter $\lambda$ to be the width $d$ of the FC layers, so that the asymptotic regime $\lambda \to \infty$ corresponds to $d \to \infty$.

\smallskip

\begin{enumerate}[leftmargin=*, nosep]
    \item \textbf{Setup.} The challenger $\mathcal{C}$ fixes an input distribution $\mathcal{D}$ and a clean-reference distribution $F_{\mathrm{clean}}$ at scale $\lambda$. We require only that $F_{\mathrm{clean}}$ be supported on functionally clean classifiers for $\mathcal{D}$; it need not coincide with the natural distribution induced by a particular training algorithm. Let $\mathsf{Atk}$ denote the (possibly randomized) backdoor attack algorithm, which takes a clean classifier as input and returns a backdoor-injected classifier.
    \item \textbf{Challenge.} $\mathcal{C}$ samples a bit $b \sim \{0,1\}$ uniformly at random and a clean classifier $f \leftarrow \mathcal{F}_{\mathrm{clean}}$, and sends to the distinguisher
    \[
        f^* := \begin{cases}
        f & \text{if } b = 0, \\
        \tilde f \leftarrow \mathsf{Atk}(f) & \text{if } b = 1.
        \end{cases}
    \]
    \item \textbf{Guess.} Given white-box access to all parameters of $f^*$, the distinguisher $\mathcal{G}$ outputs a bit $b' \in \{0,1\}$ and wins if $b' = b$.
\end{enumerate}

\noindent \textbf{Advantage.} The advantage of $\mathcal{G}$ against $\mathsf{Atk}$ is
\[
    \mathrm{Adv}^{\mathsf{Atk}}_{\mathcal{G}}(\lambda) \;=\; \left|\,2\cdot\Pr[\,b' = b\,] - 1\,\right|,
\]

where the probability is over the choice of $b$, the sampling of $f \leftarrow \mathcal{F}_{\mathrm{clean}}$, the randomness of $\mathsf{Atk}$, and the randomness of $\mathcal{G}$. \\

\noindent \textbf{Undetectability.} The attack $\mathsf{Atk}$ is \emph{undetectable} if $\mathrm{Adv}^{\mathsf{Atk}}_{\mathcal{G}}(\lambda) = o(1)$ as $\lambda \to \infty$ for every PPT distinguisher $\mathcal{G}$. That is, no PPT algorithm can distinguish $\tilde f$ from a model sampled from a distribution on functionally clean classifiers with constant advantage.

\smallskip

\noindent \textbf{Proof strategy: the clean-reference distribution.}
Establishing undetectability reduces to exhibiting two distributions of models: $\mathcal{P}_1$, the backdoor-injected distribution (correct verdict: ``backdoored''), and $\mathcal{P}_0$, the clean-reference distribution (correct verdict: ``clean''), which plays the role of $\mathcal{F}_{\mathrm{clean}}$ in the game above. Once such a pair exists, samples from $\mathcal{P}_0$ and $\mathcal{P}_1$ serve as counterexamples that defeat every PPT detector. \textbf{Crucially, $\mathcal{P}_0$ need not correspond to models produced by a benign training pipeline; it suffices that every model in $\mathcal{P}_0$ would be classified as clean by any sound detector.}

Since pretrained models lack a canonical distribution over their weights, we construct $\mathcal{P}_0$ by adding calibrated Gaussian dither around a fixed benign model $f$ (produced by a benign training pipeline). Every sample $f' \sim \mathcal{P}_0$ produces the same classification as $f$ on every input (formally established in Section~\ref{sec:undetect_proof}). Since whether a model is backdoored is determined by its classification on every input, any sound detector must assign the same verdict to two models that produce identical classifications on every input; every member of $\mathcal{P}_0$ therefore inherits $f$'s clean verdict.

Constructing $\mathcal{P}_0$ manually, rather than sampling from an intractable population of ``naturally occurring'' clean models, does not weaken the guarantee, as the argument depends only on the verdicts assigned to models. This is the standard counterexample-based impossibility argument used in cryptography.

\begin{figure}[t]
\centering
\begin{tikzpicture}[>=stealth, font=\small]
  \def\leftcol{1.5}
  \def\rightcol{5.0}

  \node[font=\normalsize\bfseries] at (\leftcol, 0) {$\mathcal{C}$};
  \node[font=\normalsize\bfseries] at (\rightcol, 0) {$\mathcal{G}$};
  \node[font=\scriptsize] at (\leftcol, -0.3) {(Challenger)};
  \node[font=\scriptsize] at (\rightcol, -0.3) {(Distinguisher)};

  \draw (\leftcol, -0.5) -- (\leftcol, -3.9);
  \draw (\rightcol, -0.5) -- (\rightcol, -3.9);

  \node[font=\scriptsize, anchor=east] at (\leftcol - 0.2, -0.85)
    {$b \sim \{0,1\}$, \; $f \leftarrow \mathcal{F}_{\mathrm{clean}}$};
  \node[font=\footnotesize\bfseries, anchor=east] at (\leftcol - 0.2, -1.15) {Step 1: Setup};

  \draw[->, thick] (\leftcol, -2.10) -- (\rightcol, -2.10)
    node[midway, above, font=\footnotesize]
      {$f^* := \begin{cases} f & b=0 \\ \mathsf{Atk}(f) & b=1 \end{cases}$}
    node[midway, below, font=\footnotesize\bfseries] {Step 2: Challenge};

  \draw[->, thick] (\rightcol, -3.05) -- (\leftcol, -3.05)
    node[midway, above, font=\footnotesize] {$b'$}
    node[midway, below, font=\footnotesize\bfseries] {Step 3: Guess};

  \draw[densely dashed, gray] (-0.3, -3.55) -- (6.8, -3.55);
  \node[font=\footnotesize\itshape] at (3.25, -3.75)
    {$\mathcal{G}$ wins if $b' = b$};
\end{tikzpicture}
\caption{\textbf{Security game for backdoor undetectability.} The challenger $\mathcal{C}$ samples a random bit $b$ and either forwards a freshly drawn clean classifier $f \sim \mathcal{F}_{\mathrm{clean}}$ or runs the attack $\mathsf{Atk}$ on $f$ and forwards the result. The distinguisher $\mathcal{G}$ is given white-box access to $f^*$ and attempts to guess $b$.}
\label{fig:security_game}
\end{figure}
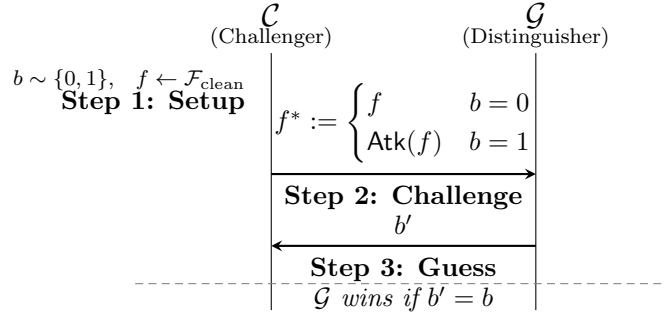

\subsection{Threat Model}

\textbf{Adversary.} We consider an adversary who distributes a pre-trained model to downstream users (e.g., through a model repository or supply chain). The adversary has full control over the model parameters and seeks to inject an undetectable backdoor such that the model maintains correct performance on clean inputs while classifying trigger-embedded inputs as a target class $y_t$. The adversary knows the training distribution $\mathcal{D}$, the model architecture, and can perform arbitrary modifications to the weights. 

\smallskip

\noindent \textbf{Defender.} The defender receives the model and has full access to its parameters. The defender is a PPT algorithm that also has access to the distribution $\mathcal{D}$. The defender may be adaptive and is assumed to know the general attack algorithm, but not the randomness used during injection, the trigger $\Delta$, or the target class $y_t$.

\section{Provably Undetectable Backdoor Attack}
\label{sec:our_attack}

We present \textit{Sparse Backdoor}, an attack that constructs a backdoor-injected classifier $\tilde{f}$ from a well-trained model $f = \{f_{enc}, W_1, \dots, W_L\}$ together with an optimized trigger $\Delta^*$. Given an adversary-chosen target class $y_t \in \mathcal{Y}$, the attack outputs a backdoored classifier $\tilde{f}$ such that $\tilde{f}(\mathcal{A}_{\Delta^*}(x)) = y_t$ for any clean input $x$ with high probability. The attack modifies only the fully connected layers of $f$, leaving the feature encoder $f_{enc}$ unchanged, and is designed so that detecting the modification in the parameters of $\tilde{f}$ is at least as hard as solving the Sparse PCA detection problem. 

\smallskip

\noindent \textbf{Overview (Figure~\ref{fig:attack_pipeline}).} The central mechanism is a sparse signal that is injected at the boundary between the feature encoder and fully connected stages and then propagated layer by layer to the output. An input-space trigger $\Delta^*$ is optimized so that, when added to any clean image $x$, the frozen feature encoder $f_{enc}$ produces an embedding with an increased component along a randomly chosen sparse direction $s_1$; this constitutes the entry point of the backdoor signal to the FC layers. At each FC layer $i$, a randomly chosen sparse direction $s_i \in \mathbb{R}^{d_i}$ serves as the carrier of the backdoor signal. The attack perturbs a small subset of columns of the weight matrix $W_i$ by adding noise aligned with $s_i$; these perturbations selectively amplify the component along $s_i$ in the layer's input and relay it to a new sparse direction $s_{i+1}$ in the next layer's input space. At the final layer, the strongest sampled perturbation is assigned to the column corresponding to the target class $y_t$, so that the accumulated signal produces the targeted misclassification. For clean inputs, which carry no significant component along $s_1$, the projection $s_1^\top x \approx 0$ makes the rank-one perturbation $\xi_1 s_1 s_1^\top$ act trivially, so the perturbed model's clean-input behavior is unchanged. 

To mask these structured sparse perturbations, each modified column also receives an independent dense Gaussian dither. A clean model augmented with only this dither serves as the reference distribution: distinguishing the backdoor-injected weights from this reference requires recovering the hidden sparse direction, which is precisely the Sparse PCA detection. We formalize this reduction and prove the resulting undetectability guarantee in Section~\ref{sec:undetect_proof}. 

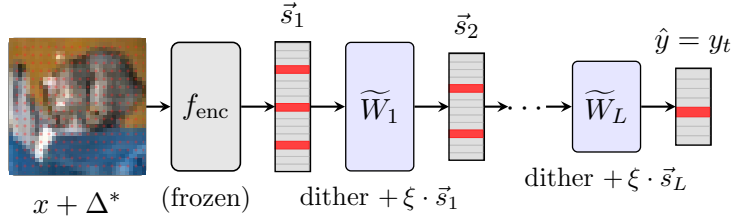
\begin{figure}[t]
\centering
\begin{tikzpicture}[
  every node/.style={font=\small},
  enc/.style={draw, rounded corners=3pt, fill=gray!20, line width=0.6pt, align=center},
  fc/.style={draw, rounded corners=2pt, fill=blue!10, line width=0.5pt, align=center},
  emb/.style={draw, line width=0.5pt, inner sep=0pt},
  arrow/.style={->, >=stealth, thick},
  bg/.style={fill=gray!25, draw=gray!50, line width=0.2pt},
  trig/.style={fill=red!75, draw=red!85, line width=0.2pt},
]
\node[draw, rounded corners=2pt, line width=0.6pt, inner sep=0pt] (img) at (0, 0)
  {\includegraphics[width=1.8cm,height=1.8cm]{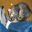}};
\fill[pattern=dots, pattern color=red!85, opacity=0.6] ([xshift=0.04cm,yshift=0.04cm]img.south west) rectangle ([xshift=-0.04cm,yshift=-0.04cm]img.north east);
\node[font=\small, anchor=north] at ([yshift=-0.08cm]img.south) {$x + \Delta^*$};

\node[enc, minimum width=0.85cm, minimum height=1.7cm] (enc) at (1.70, 0) {$f_{\mathrm{enc}}$};
\node[font=\footnotesize, anchor=north] at ([yshift=-0.05cm]enc.south) {(frozen)};

\node[emb, minimum width=0.45cm, minimum height=1.7cm] (e1) at (2.85, 0) {};
\foreach \i in {0,...,16} \fill[bg]   ([yshift=0.1*\i cm]e1.south west) rectangle ++(0.45, 0.1);
\foreach \i in {3, 8, 13} \fill[trig] ([yshift=0.1*\i cm]e1.south west) rectangle ++(0.45, 0.1);
\draw[line width=0.5pt] (e1.south west) rectangle (e1.north east);
\node[font=\small, anchor=south] at ([yshift=0.05cm]e1.north) {$\vec{s}_1$};

\node[fc, minimum width=0.90cm, minimum height=1.7cm] (w1) at (4.00, 0) {$\widetilde{W}_1$};
\node[font=\footnotesize, anchor=north] at ([yshift=-0.05cm]w1.south) {dither $+\,\xi\cdot\vec{s}_1$};

\node[emb, minimum width=0.45cm, minimum height=1.4cm] (e2) at (5.15, 0) {};
\foreach \i in {0,...,13} \fill[bg]   ([yshift=0.1*\i cm]e2.south west) rectangle ++(0.45, 0.1);
\foreach \i in {3, 9}     \fill[trig] ([yshift=0.1*\i cm]e2.south west) rectangle ++(0.45, 0.1);
\draw[line width=0.5pt] (e2.south west) rectangle (e2.north east);
\node[font=\small, anchor=south] at ([yshift=0.05cm]e2.north) {$\vec{s}_2$};

\node[font=\normalsize] at (5.95, -0.03) {$\cdots$};

\node[fc, minimum width=0.90cm, minimum height=1.2cm] (wL) at (7.00, 0) {$\widetilde{W}_L$};
\node[font=\footnotesize, anchor=north] at ([yshift=-0.05cm]wL.south) {dither $+\,\xi\cdot\vec{s}_L$};

\node[emb, minimum width=0.45cm, minimum height=1.0cm] (lg) at (8.15, 0) {};
\foreach \i in {0,...,7} \fill[bg]   ([yshift=0.125*\i cm]lg.south west) rectangle ++(0.45, 0.125);
\fill[trig] ([yshift=0.375cm]lg.south west) rectangle ++(0.45, 0.125);
\draw[line width=0.5pt] (lg.south west) rectangle (lg.north east);
\node[font=\small, anchor=south] at ([yshift=0.05cm]lg.north) {$\hat{y}=y_t$};

\draw[arrow] (img.east) -- (enc.west);
\draw[arrow] (enc.east) -- (e1.west);
\draw[arrow] (e1.east) -- (w1.west);
\draw[arrow] (w1.east) -- (e2.west);
\draw[arrow] (e2.east) -- (5.7, 0);
\draw[arrow] (6.20, 0) -- (wL.west);
\draw[arrow] (wL.east) -- (lg.west);
\end{tikzpicture}
\caption{\textbf{Sparse Backdoor pipeline.} The trigger $\Delta^*$ (red dots) blended into input $x$ is passed through a frozen feature encoder $f_{\mathrm{enc}}$, producing an embedding with a high component along a sparse direction $s_1$ (red coordinates). Each perturbed FC layer $\widetilde{W}_i$ combines a Gaussian dither with a structured spike along $s_i$, propagating the signal to a sparser direction $s_{i+1}$ as embeddings shrink with depth. The final layer routes the spike to the target-class column $y_t$ in the output logits.}
\label{fig:attack_pipeline}
\end{figure}
\smallskip

\noindent The attack follows a meet-in-the-middle strategy~\cite{hong2022handcrafted} with three stages: (i) \textbf{Trigger Optimization} finds an input perturbation $\Delta^*$ that drives the frozen feature encoder $f_{enc}$ to produce an embedding with a large component along $s_1$; (ii) \textbf{Intermediate Injection} adds structured sparse perturbations to each hidden FC layer to propagate the signal from $s_i$ to $s_{i+1}$; and (iii) \textbf{Final Injection} perturbs the last FC layer to route the signal to the target class $y_t$. Each stage is detailed in Section~\ref{sec:attack-construction}. The attack makes no assumptions about $f_{enc}$ and therefore applies to any architecture terminating in FC layers, provided the adversary can optimize a trigger that activates the backdoor signal in the embeddings.

\subsection{Sparse Signal Model}
\label{sec:sparse-signal}
This subsection defines the signal model underlying our attack. We describe the sparse directions that carry the backdoor signal, the structure that enables its propagation, and the Gaussian dithering that enables the reduction from Sparse PCA in Section~\ref{sec:undetect_proof}. 

\smallskip

\noindent \textbf{Why Sparse directions? } The choice of sparse directions is motivated by two considerations. First, propagating the signal to the next layer requires producing a component along a $k_{i+1}$-sparse direction $s_{i+1}$, which involves modifying only $O(k_{i+1})$ columns of the weight matrix $W_i$, since each column controls a single neuron. This locality limits the number of weight modifications per layer and thereby limits the impact on clean accuracy. Second, restricting perturbations to a sparse direction means that the resulting deviation is concentrated along an unknown sparse subspace. Distinguishing such a perturbation from isotropic noise reduces to the Sparse PCA detection problem, which underpins our undetectability guarantee. 

\smallskip

\noindent \textbf{Sparse Backdoor directions. } At each FC layer $i$ with input dimension $d_i$, we associate a randomly chosen $k_i$-sparse unit vector $s_i \in \mathbb{R}^{d_i}$ satisfying $\lVert s_i \rVert_0 = k_i$ and $\lVert s_i\rVert_2 = 1$, where $k_i  = \lfloor d_i^\alpha \rfloor$ for some $\alpha \in (0, 1/2]$. The vector $s_i$ defines the direction along which the backdoor signal is embedded at layer $i$. We say that an intermediate representation $x_i^*$ carries the backdoor signal if it has a significantly larger projection onto $s_i$ than the corresponding clean representation $x_i$. Concretely, we model a corrupted representation 
\[
    x_i^* = x_i + \beta_i \cdot s_i
\]
where $\beta_i > 0$ quantifies the signal strength at layer $i$, so that $\langle x_i^*, s_i\rangle  - \langle x_i, s_i\rangle \geq \beta_i$. For clean inputs, which have no systematic alignment with the randomly chosen $s_i$, we assume that the projection $\langle x_i, s_i\rangle$ is small relative to $\beta_i$. Since $s_i$ is chosen independently of the trained model and is supported on only $k_i$ out of $d_i$ coordinates, a typical embedding that distributes its mass over many dimensions has a projection of order $\lVert x_i\rVert_2 \cdot \sqrt{k_i/d_i}$, which shrinks as $d_i$ grows. We formalize the requirement that embeddings spread this mass across coordinates as the orthogonality condition in Section~\ref{sec:correctness} and verify it empirically in Appendix~\ref{app:theorem61-verification}.  

\smallskip

\noindent \textbf{Backdoor perturbation.} To enable signal propagation, the attack modifies a subset of columns of each FC layer weight matrix $W_i$. Since each column of $W_i$ controls a single neuron in the next layer, and the next-layer sparse direction $s_{i+1}$ has sparsity $k_{i+1}$, at least $k_{i+1}$ columns must be modified to produce a signal with the desired support. We sample $|\mathcal{I}_i| = \lfloor c \cdot k_{i+1} \rfloor$ column indices uniformly at random from $\{0, \dots, d_{i+1} - 1\}$, where $c \geq 1$ is a small constant oversampling factor whose role will become clear shortly.

For each selected column $j\in \mathcal{I}_i$, the perturbation takes the form $\xi_i^{(j)} \cdot s_i$, where the scalar backdoor coefficient $\xi_i^{(j)} \sim \mathcal{N}(0, \sigma_i^2)$ controls the magnitude and sign of the perturbation along $s_i$. The perturbed column becomes 
\[
    \tilde{w}_i^{(j)} = w_i^{(j)} + \xi_i^{(j)} \cdot s_i.
\]
Since the perturbation is aligned with $s_i$, it selectively amplifies the backdoor component in the layer's input: a corrupted input $x_i^*$ with a large projection onto $s_i$ produces an activation shift proportional to $\xi_i^{(j)}\cdot \beta_i$ at neuron $j$, while a clean input with negligible projection onto $s_i$ is largely unaffected. Columns with positive coefficients $\xi_i^{(j)} > 0$ contribute positive activation shifts that can survive the ReLU nonlinearity, and the indices of the $k_{i+1}$ largest such coefficients define the support of the next-layer sparse direction $s_{i+1}$. The oversampling factor $c$ ensures that, with high probability, at least $k_{i+1}$ positive coefficients are available for this selection. 

\smallskip

\noindent \textbf{Gaussian dither.} The structured perturbations described above are sufficient to inject and propagate the backdoor signal. To enable a clean reduction from Sparse PCA, we additionally apply an independent dense Gaussian perturbation $\eta_i^{(j)} \sim \mathcal{N}(0, \tau_i^2 \cdot I_{d_i})$ to each modified column. This gives rise to two column-level objects: a \emph{dither-only} column
$ w_i'^{(j)} \;=\; w_i^{(j)} + \eta_i^{(j)},$
which defines the reference model $f'$, and the \emph{full perturbed} column
\[
    \tilde{w}_i^{(j)} \;=\; w_i^{(j)} + \eta_i^{(j)} + \xi_i^{(j)} \cdot s_i \;=\; w_i'^{(j)} + \xi_i^{(j)} \cdot s_i
\]
which defines the backdoored model $\tilde{f}$. The only difference between $\tilde{f}$ and $f'$ is the structured spike $\xi_i^{(j)} \cdot s_i$ at each modified column.

The dither plays a purely technical role inside the proof of undetectability: the model $f'$ obtained by applying \emph{only} the Gaussian dither to $f$ serves as a clean reference against which undetectability is proven, allowing the detection to be reduced from Sparse PCA. Crucially, we \emph{prove} this reference model to be clean: under a mild margin condition on the baseline $f$ and a calibrated choice of $\tau_i$, Lemmas~\ref{lem:out-stab} and~\ref{lem:dither-no-bd} establish that $f'$ computes the same function as $f$ on every input, and thus cannot exhibit any backdoor behavior.

The variance $\tau_i^2$ is calibrated so that the propagated logit perturbation from the dither stays within the classification margin of $f$ (Assumption~\ref{assm:dither}; see Section~\ref{sec:undetect_proof} for operational interpretation). The attack parameters are further calibrated so that
\[
    \sigma_i^2 \;=\; o\!\left(\tau_i^2 \cdot k_i \cdot \sqrt{\frac{\log d_i}{|\mathcal{I}_i|}}\right),
\]
placing the resulting weight distribution in the Sparse PCA detection hardness regime: $\sigma_i^2 = \Omega(\tau_i^2)$ ensures that the structured perturbation is large enough to reliably propagate the signal, yet small enough relative to the isotropic dither that no polynomial-time algorithm can distinguish $\tilde{f}$ from $f'$ with a constant advantage.

\subsection{Attack Construction}
\label{sec:attack-construction}
We now describe how each stage is realized. Using the perturbation structure from Section~\ref{sec:sparse-signal}, the attack must satisfy three requirements for the sparse signal to produce a targeted misclassification: 

\smallskip

\noindent \textbf{(1) Trigger activation. } The optimized trigger $\Delta^*$ must induce a strong component along $s_1$ at the input to the first FC layer. For any clean input image $x$, 
\begin{equation}
\label{eq:trigger-activation}
\langle f_{{enc}}(\mathcal{A}_{\Delta^*}(x)),\, s_1 \rangle
\;-\;
\langle f_{{enc}}(x),\, s_1 \rangle
\;\geq\; \beta_1.
\end{equation}

\noindent \textbf{(2) Signal propagation. } Each hidden FC layer $i \in \{1, \dots, L-1\}$ must relay the signal from $s_i$ to $s_{i+1}$. For a corrupted input $x^*_i = x_i + \beta_i\cdot s_i$, 
\begin{equation}
\label{eq:signal-propagation}
\langle \mathrm{ReLU}(\tilde{W}_i^\top x_i^*),\, s_{i+1} \rangle
\;-\;
\langle \mathrm{ReLU}(\tilde{W}_i^\top x_i),\, s_{i+1} \rangle
\;\geq\; \beta_{i+1}.
\end{equation}

\noindent \textbf{(3) Targeted misclassification. } At the final FC layer, the accumulated signal must produce the target prediction: 
\begin{equation}
\label{eq:target-misclassification}
\arg\max \; \operatorname{Softmax}(\tilde{W}_L^\top x_L^*) = y_t.
\end{equation}

\noindent We now describe the algorithm for each stage. 

\smallskip

\noindent \textbf{Trigger optimization.} Alg.~\ref{alg:trigger_opt} (\texttt{Trigger-Optimization}) optimizes an input-space trigger satisfying the trigger activation condition (Eq.~\ref{eq:trigger-activation}). The algorithm begins by randomly selecting a subset $\mathcal{T} \subset \{0, \dots, d_1 - 1\}$ of size $k_1$, which defines the support of the desired sparse direction. Using the frozen feature encoder $f_{enc}$ and a set of clean sample images $X$, it iteratively optimizes a trigger $\Delta$ via gradient descent. At each iteration, the trigger is constrained to lie in $[-\delta, \delta]$ using a $\operatorname{tanh}$ reparameterization. The optimization objective maximizes the average activation shift on the target coordinates $\mathcal{T}$ while penalizing leakage onto the remaining coordinates via a quadratic regularizer, encouraging a sparse shift concentrated on the desired support. After convergence, the induced backdoor direction is estimated as the average embedding shift restricted to $\mathcal{T}$ and normalized to unit norm, yielding the $k_1$-sparse direction $s_1$. The full pseudocode is given in Appendix~\ref{app:algorithm-details} (Alg.~\ref{alg:trigger_opt}).  

\smallskip

\noindent \textbf{Intermediate injection. } Alg.~\ref{alg:backdoor_inj} (\texttt{Mid-Injection}) realizes the signal propagation condition (Eq.~\ref{eq:signal-propagation}) for each hidden FC layer $i \in \{1, \dots, L -1\}$. Given the weight matrix $W_i$ and the current-layer sparse direction $s_i$, the algorithm samples $|\mathcal{I}_i| = \lfloor c \cdot k_{i+1}\rfloor$ candidate column indices and perturb each selected column $j$ by adding the Gaussian dither $\eta_i^{(j)}$ and the structured backdoor perturbation $\xi_i^{(j)}\cdot s_i$.
The algorithm then collects all the indices with positive backdoor coefficients and selects the $k_{i+1}$ largest to define the support of $s_{i+1}$, which is normalized to a unit vector. This procedure ensures that the columns contributing the strongest positive activation shifts under the backdoor signal determine the next-layer sparse direction, enabling reliable signal propagation across layers. See Appendix~\ref{app:algorithm-details} (Alg.~\ref{alg:backdoor_inj}) for the full pseudocode.

\smallskip

\noindent \textbf{Final injection. } Alg.~\ref{alg:final_backdoor_inj} (\texttt{Final-Injection}) realizes the misclassification condition (Eq.~\ref{eq:target-misclassification}). Unlike the intermediate layers, the final FC layer maps directly to the output logits, so the goal is not to propagate the signal further but to ensure that the target class $y_t$ receives the largest activation shift. The algorithm samples independent backdoor coefficients for all columns from $\mathcal{N}(0, \sigma_L^2)$ and assigns the largest positive coefficient to the column corresponding to $y_t$, with the remaining coefficients randomly assigned to the other classes. Each column then receives the Gaussian dither and structured perturbation as in the intermediate layers. Since the activation shift at each output neuron is proportional to the corresponding backdoor coefficient, this assignment ensures that the target-class logit receives the maximum positive shift when the backdoor signal is present, producing the targeted misclassification. To minimize the impact on clean performance, the adversary may perturb only a subset of columns, provided that the target-class column is included. See Appendix~\ref{app:algorithm-details} (Alg.~\ref{alg:final_backdoor_inj}) for the full pseudocode.

\smallskip

\noindent \textbf{Sparse Backdoor Attack. } Alg.~\ref{alg:full_attack} integrates the three stages into the end-to-end Sparse Backdoor attack. Starting from a clean pretrained classifier $f$, the attack first optimizes the trigger $\Delta^*$ and extracts the initial sparse direction $s_1$. It then sequentially applies intermediate injection to each hidden FC layer, propagating the signal through directions $s_1, \dots, s_L$. Finally, it applies a final injection to route the signal to the target class $y_t$. The resulting classifier $\tilde{f} =\{f_{enc}, \tilde{W}_1,\dots, \tilde{W}_L\}$ and the trigger $\Delta^*$ constitute the full attack. We prove the correctness of the signal propagation mechanism and the undetectability of the resulting weight distributions in Section~\ref{sec:theory}. 

\smallskip

\begin{algorithm}[!ht]
\caption{\texttt{Sparse Backdoor Attack}}
\label{alg:full_attack}
\KwIn{Clean classifier $f = \{f_{enc}, W_1, \dots, W_L\}$; clean samples $X$; sparsity targets $\{k_1, \dots, k_{L-1}\}$; oversampling factor $c$; attack parameters $\{\sigma_i^2, \tau_i^2\}_{i=1}^L$; trigger bound $\delta$; target class $y_t$}
\KwOut{Backdoored classifier $\tilde{f}$; optimal trigger $\Delta^*$}
\tcp{Phase 1: optimize trigger, returning normalized $k_1$-sparse direction $s_1$}
$\Delta^*, s_1 \leftarrow \text{\texttt{Trigger-Optimization}}(f_{enc}, X, k_1, \delta)$ ~(Alg.~\ref{alg:trigger_opt})\;
\tcp{Phase 2: propagate signal through intermediate layers, picking top-$k_{i+1}$ direction at each step}
\For{$i = 1$ \KwTo $L-1$}{
    $\tilde{W}_i, s_{i+1} \leftarrow \text{\texttt{Mid-Injection}}(W_i, s_i, k_{i+1}, c, \sigma_i^2, \tau_i^2)$ ~(Alg.~\ref{alg:backdoor_inj})\;
}
\tcp{Phase 3: route the strongest positive signal to the target class $y_t$}
$\tilde{W}_L \leftarrow \text{\texttt{Final-Injection}}(W_L, s_L, y_t, \sigma_L^2, \tau_L^2)$~(Alg.~\ref{alg:final_backdoor_inj})\;
\Return $\tilde{f} = \{f_{enc}, \tilde{W}_1, \dots, \tilde{W}_L\}, \Delta^*$\;
\end{algorithm}

\noindent \textbf{Practical extensions.} Two extensions improve effectiveness while preserving undetectability. \emph{(1) Basis-aligned trigger:} maximize trigger shift energy across all $d_1$ coordinates and apply an orthonormal basis change to make the dominant direction $k_1$-sparse. Since Gaussian noise is rotationally invariant and SPCA hardness holds in any orthonormal basis, a distinguisher for the rotated construction would imply one for canonical SPCA. \emph{(2) Targeted competitor suppression:} in Alg.~\ref{alg:final_backdoor_inj}, sample non-target coefficients i.i.d.\ $\sim \mathcal{N}(0, \sigma_L^2)$ but assign the most negative ones to classes whose clean weights are most responsive to $s_L$. Only the class--coefficient pairing changes, so each column's marginal distribution is unchanged; see Appendix~\ref{app:extensions}.

\section{Theoretical Analysis}
\label{sec:theory}
We establish two core guarantees: correctness of signal propagation (Section~\ref{sec:correctness}) and computational undetectability (Section~\ref{sec:undetect_proof}).

\smallskip

\noindent \textbf{Challenges in the pretrained setting.}
Goldwasser et al.'s construction~\cite{goldwasser2022planting} is limited to single-layer random-weight networks, and extending it to multi-layer pre-trained classifiers introduces three challenges. First, the reference distribution of pre-trained weights is not known in closed form, which we address by introducing Gaussian dithering to construct the distribution over the clean reference models $f'$ centered at the baseline model $f$ (Section~\ref{sec:undetect_proof}). Second, perturbations layered on top of pretrained weights risk interfering with the clean computations those weights already encode, which motivates the mild orthogonality and non-degeneracy assumptions of Section~\ref{sec:correctness}. Third, multi-layer propagation creates cross-layer dependencies on the sparse direction, which we handle via a reduction from Sparse PCA, together with a hybrid argument extending the single-layer guarantee to $L$ layers (Section~\ref{sec:undetect_proof}).

\subsection{Correctness}
\label{sec:correctness}

We show that each perturbed fully connected layer $\tilde{W}_i$ from Algorithm~\ref{alg:backdoor_inj} correctly propagates the backdoor signal: if the input to layer $i$ carries signal strength $\beta_i$ along $s_i$, the expected post-ReLU output shifts positively along $s_{i+1}$, satisfying Eq.~\ref{eq:signal-propagation}. Iterating this step guarantees reliable transmission of the trigger signal through successive FC layers to the target class.

\begin{theorem}[\textbf{Correctness of Signal Propagation}]\label{thm:activation}
    Fix an FC layer index $i$ and let $\tilde{W}_i$ and $s_{i+1}$ be produced from $W_i$ by Algorithm~\ref{alg:backdoor_inj} with candidate set $\mathcal{I}_i \subseteq [d_{i+1}]$ and input signal $s_i \in \mathbb{R}^{d_i}$ with $\lVert s_i\rVert_2 = 1$. Recall that for each $j \in \mathcal{I}_i$, the algorithm samples a Gaussian dither $\eta_i^{(j)} \sim \mathcal{N}(0, \tau_i^2I_{d_i})$ and a backdoor coefficient $\xi_i^{(j)} \sim \mathcal{N}(0, \sigma_i^2)$. Let $\eta := \{\eta_i^{(j)}\}_{j \in \mathcal{I}_i}$ and $\xi := \{\xi_i^{(j)}\}_{j \in \mathcal{I}_i}$ denote the respective collections.
    
    Let $x_i \in \mathbb{R}^{d_i}$ be a fixed clean feature vector and $x_i^* = x_i + \beta_i \cdot s_i$ the corresponding poisoned input for some $\beta_i \geq 0$. Assume:
    \begin{enumerate}
        \item \textbf{(Orthogonality)} The sparse backdoor direction $s_i$ is orthogonal to both the clean feature vector and the clean weight vectors: $\langle s_i, x_i \rangle = 0$ and $\langle w_i^{(j)}, s_i \rangle = 0$ for each $j \in \mathcal{I}_i$.
        \item \textbf{(Non-degeneracy)} Each neuron $j \in \mathcal{I}_i$ has a non-negligible probability of being active on clean inputs: $P_{\eta}(Z_j(x_i) > 0) \ge c_0$ for some constant $c_0 > 0$, where $Z_j(x_i) := \langle \tilde{w}_i^{(j)}, x_i \rangle$.
    \end{enumerate}
    Then, for the backdoor coefficients $\xi$, the following lower bound holds:
    \begin{multline*}
    \Big\langle
    \mathbb{E}_{\eta}\big[\mathrm{ReLU}(\tilde{W}_i^\top x_i^*)\big]
    -
    \mathbb{E}_{\eta}\big[\mathrm{ReLU}(\tilde{W}_i^\top x_i)\big],
    \ s_{i+1}
    \Big\rangle \\
    \ge
    c_0 \cdot \beta_i \cdot \frac{1}{\sqrt{k_{i+1}}} \cdot\sum_{j:\xi_i^{(j)}>0}\xi_i^{(j)}
    \ -\ \sqrt{k_{i+1}}\,\cdot \beta_i \cdot \tau_i \cdot\sqrt{\frac{2}{\pi}},
    \end{multline*}
    where $k_{i+1} = \|v\|_0$ with $v[j] = \mathbf{1}\{\xi_i^{(j)} > 0\}$. Taking expectation over $\xi$ yields a directional signal gain of order $\Omega(\beta_i \cdot \sigma_i)$ along $s_{i+1}$.
\end{theorem}
\begin{proof}[\textbf{Proof Sketch.}]
    The proof proceeds in three steps. Fix an output neuron $j \in \mathcal{I}_i$. Under orthogonality, the pre-ReLU difference between the poisoned and clean inputs decomposes cleanly into the backdoor signal $\beta_i \cdot \xi_i^{(j)}$ and a zero-mean dither interference term $\beta_i \cdot \langle \eta_i^{(j)}, s_i \rangle$. The key challenge is to show that this signal survives the nonlinearity of ReLU. 

    \smallskip

    \noindent The central observation is that ReLU acts linearly on positive inputs. When a neuron is already active on the clean input ($Z_j(x_i) > 0$), both the clean and shifted pre-activations lie in the linear regime, and the entire signal $\beta_i \cdot \xi_i^{(j)}$ passes through unchanged. The non-degeneracy assumption guarantees this happens with probability at least $c_0$, yielding a per-neuron gain of $c_0\cdot\beta_i\cdot\xi_i^{(j)}$. The dither interference, while also passing through ReLU, is zero-mean and contributes only through its expected magnitude $\beta_i\cdot \tau_i \cdot \sqrt{2/\pi}$, which can be bounded via the $1$-Lipschitz property of ReLU. 

    \smallskip
    
    \noindent Finally, Algorithm~\ref{alg:backdoor_inj} constructs $s_{i+1}$ by selecting exactly the coordinates where $\xi_i^{(j)} > 0$, ensuring that the per-neuron gains aggregate coherently when projected onto $s_{i+1}$. The dither cost, being coordinate-independent, accumulates a factor of $\sqrt{k_{i+1}}$. Taking expectation over $\xi$ confirms a directional gain of order $\Omega(\beta_i\sigma_i)$.
\end{proof}
The full proof, which formalizes the three steps above (decomposition, per-coordinate signal bound, and aggregation across the support of $s_{i+1}$), is deferred to Appendix~\ref{app:proof-thm-activation}. 

\smallskip

\noindent \textbf{Discussion of Assumptions.}
The orthogonality conditions $\langle s_i, x_i \rangle = 0$ and $\langle w_i^{(j)}, s_i\rangle = 0$ simplify the decomposition in the proof but are not essential. Writing $\epsilon := \mathrm{max}_j |\langle w_i^{(j)}, s_i\rangle|/ \lVert w_i^{(j)}\rVert_2$, the proof carries through with an additive per-layer error of $O(\beta_i\cdot \epsilon \cdot \sqrt{k_{i+1}})$, which does not dominate the signal whenever $\sigma_i = \omega(\tau_i + \epsilon)$. A random $k_i$-sparse unit vector's inner product against any fixed reference concentrates at the rate $O(\sqrt{k_i/d_i})$~\cite{vershynin2018high}, so $\epsilon$ vanishes asymptotically, and the additive error is dominated by the signal term.  Non-degeneracy holds whenever candidate neurons fire with at least constant probability on clean inputs, a mild condition for a well-trained network. We defer detailed justification and empirical verification across all nine configurations to Appendix~\ref{app:theorem61-verification}. 

\smallskip

\noindent \textbf{Implications.}
Theorem~\ref{thm:activation} shows the signal propagates through each FC layer: whenever $\beta_i > 0$, the expected post-ReLU activations drift positively along $s_{i+1}$, carrying the trigger to the target class. Since the per-neuron gain is proportional to $\xi_i^{(j)}$, the final step (Algorithm~\ref{alg:final_backdoor_inj}) maximizes the target shift by assigning the largest coefficient to $y_t$. Signal $\beta_1$ depends on $f_{\mathrm{enc}}$ and input; a smaller $\beta_1$ requires larger perturbations, which may lead to variance in ASR.

\subsection{Undetectability}
\label{sec:undetect_proof}

We now show that the sparse backdoor signal is hidden from any efficient observer. The argument has three pieces:

\begin{enumerate}[leftmargin=*, nosep]
    \item[(i)] \textbf{output stability} (Lemma~\ref{lem:out-stab}): adding Gaussian dither to $f$ yields a model $f'$ whose output is close to $f(x)$ on every input;
    \item[(ii)] \textbf{no backdoor} (Lemma~\ref{lem:dither-no-bd}): under a margin condition on $f$, $f'$ has the same argmax as $f$ on every input, natural or triggered; 
    \item[(iii)] \textbf{computational indistinguishability} (Theorem~\ref{thm:red}): distinguishing $f'$ from the backdoor-injected model $\tilde{f}$ is as hard as Sparse PCA detection. 
\end{enumerate}

Together, (i)--(ii) certify that $f'$ is a legitimate clean reference model under the functional cleanliness definition of Section~\ref{sec:formal_attack}. Therefore, the Gaussian-dithered reference distribution is a valid instantiation of $F_{\mathrm{clean}}$: it is supported on classifiers that compute the same task function as the original clean model and do not introduce any attacker-planted trigger behavior. Theorem~\ref{thm:red} then establishes undetectability in the sense of Section~\ref{section:undetectability}: no PPT detector can separate $\tilde{f}$ from this clean-reference distribution.

\smallskip

\smallskip

\noindent \textbf{Clean Reference Distribution $\mathcal{P}_0$.}
Following the proof strategy of Section~\ref{section:undetectability}, $\mathcal{P}_0$ is obtained by applying only the Gaussian dither of Section~\ref{sec:sparse-signal} to the weights of $f$, omitting the sparse backdoor perturbation. Lemmas~\ref{lem:out-stab} and~\ref{lem:dither-no-bd} below establish that every sample $f' \sim \mathcal{P}_0$ computes the same function as $f$ on every input, so $\mathcal{P}_0$ meets the sufficient condition for a clean-reference distribution stated in Section~\ref{section:undetectability}.

\begin{assumption}[\textbf{Margin Regularity}]
\label{assm:margin}
The baseline model $f$ has classification margin at least $\gamma > 0$ on a $1-o(1)$ fraction of the input distribution $\mathcal{D}$; that is,
\[
    \Pr_{x\sim\mathcal{D}}\!\left[\,g(x)_{f(x)} - \max_{y \ne f(x)} g(x)_y \;\ge\; \gamma\,\right] \;=\; 1 - o(1).
\]
\end{assumption}

\begin{assumption}[\textbf{Calibrated Dither}]
\label{assm:dither}
For each perturbed layer $i \in [L]$, the dither variance $\tau_i^2$ is calibrated so that the propagated perturbation bound of Lemma~\ref{lem:out-stab} is strictly less than $\gamma/2$ with probability at least $1 - o(1)$ over the dither randomness.
\end{assumption}

Assumption~\ref{assm:dither} is purely local: it depends only on the per-layer dither variances, layer dimensions, and operator norms of $f$'s weights, with no distributional prior over clean weights. Appendix~\ref{app:clean-ref-verification} empirically verifies both assumptions and the conclusion of Lemma~\ref{lem:dither-no-bd} on all nine (architecture, dataset) configurations.

\begin{lemma}[\textbf{Output Stability under Gaussian Dither}]
\label{lem:out-stab}
Let $f$ be an $L$-layer ReLU network with weight matrices $W_1,\ldots,W_L$ and logit map $g$, and let $f'$ (with logit map $g'$) be the model obtained by replacing each perturbed column $w_i^{(j)}$ ($j \in \mathcal{I}_i$) with $w_i^{(j)} + \eta_i^{(j)}$, $\eta_i^{(j)} \sim \mathcal{N}(0,\tau_i^2 I_{d_i})$. Then for every input $x$, with probability at least $1-\delta$ over the dither,
\[
    \|g'(x) - g(x)\|_2
    \;\le\; 2\cdot\|f_{\mathrm{enc}}(x)\|_2 \cdot \sum_{i=1}^L \tau_i \sqrt{|\mathcal{I}_i| + 2\log(L/\delta)} \cdot \prod_{j \ne i}\|W_j\|_{\mathrm{op}},
\]

\end{lemma}

\noindent The proof sketch is deferred to Appendix~\ref{app:proof-lem-out-stab}.

\begin{lemma}[\textbf{Dither Preserves Predictions}]
\label{lem:dither-no-bd}
Under Assumptions~\ref{assm:margin} and~\ref{assm:dither}, with probability $1 - o(1)$ over the dither, $f'(x) = f(x)$ for every input $x$ on which $g$ has margin $\ge \gamma$. In particular:
\begin{enumerate}[leftmargin=*, labelindent=0pt]
    \item[(a)] $f'$ matches the clean accuracy of $f$ up to an $o(1)$ additive term.
    \item[(b)] For any trigger pattern $t$ and any target class $y^\star$, the fraction of inputs $x$ for which $f'(x) \ne y^\star$ but $f'(x \oplus t) = y^\star$ is at most $o(1)$ larger than the corresponding fraction for $f$. Since the clean baseline $f$ is not backdoored, neither is $f'$.
\end{enumerate}
\end{lemma}
\noindent The proof sketch is deferred to Appendix~\ref{app:proof-lem-dither-no-bd}.

Lemma~\ref{lem:dither-no-bd} is the structural content we need: Gaussian dither cannot \emph{create} a backdoor, because it cannot change the predicted label of $f$ on any input. The dithered model $f'$ therefore inherits the functional behavior of $f$, including the absence of any trigger-activated response. This certifies $f'$ as a clean reference model. Appendix~\ref{app:clean-ref-verification} reports the empirical per-sample agreement $\Pr[f'(x)=f(x)]$, the mean margin, and the Lemma~\ref{lem:out-stab} certification rate for all 9 configurations, giving empirical backing to this structural claim.

\begin{definition}[\textbf{Shifted Sparse PCA Weight Distributions}]
\label{def:weight-dist}
Fix a weight matrix $W_i \in \mathbb{R}^{d_i \times d_{i+1}}$ and a candidate set $\mathcal{I}_i \subseteq [d_{i+1}]$. For each $j \in \mathcal{I}_i$, define the weight column distributions of the clean reference model $f'$ and the backdoor-injected model $\tilde{f}$ as follows:
\[
\begin{aligned}
\text{Clean Model } (f'): \quad
& w_i'^{(j)} = w_i^{(j)} + \eta_i^{(j)}, \\
\text{Backdoor Model } (\tilde{f}): \quad
& \tilde{w}_i^{(j)} = w_i^{(j)} + \eta_i^{(j)} + \xi_i^{(j)} \cdot s_i.
\end{aligned}
\]
where $\eta_i^{(j)} \sim \mathcal{N}(0, \tau_i^2 I_{d_i})$, $\xi_i^{(j)} \sim \mathcal{N}(0, \sigma_i^2)$, and $s_i \in \mathbb{R}^{d_i}$ is a $k_i$-sparse unit vector unknown to the detector.

Equivalently, after subtracting the deterministic offset $w_i^{(j)}$ and rescaling by $1/\tau_i$, distinguishing between $f'$ and $\tilde{f}$ reduces to deciding whether $|\mathcal{I}_i|$ i.i.d.\ samples were drawn from
\[
\mathcal{N}(0, I_{d_i})
\quad\text{or}\quad
\mathcal{N}\!\left(0, I_{d_i}
+ \theta \cdot s_i s_i^\top\right),
\qquad
\theta := \frac{\sigma_i^2}{\tau_i^2},
\]

which is exactly the Sparse PCA detection problem. We refer to this as the \emph{shifted Sparse PCA detection problem}.
\end{definition}

\begin{theorem}[\textbf{Hardness of Detection}]
\label{thm:red}
Let $f'$ and $\tilde{f}$ be the clean reference and backdoor-injected models defined in Definition~\ref{def:weight-dist}, with $L$ fully connected layers perturbed independently and parameters at each layer satisfying $\theta_i = \sigma_i^2/\tau_i^2 = o\!\left(k_i \cdot \sqrt{\log d_i / |\mathcal{I}_i|}\right)$. Under Assumption~\ref{assm:hardness} and for constant $L$, the Sparse Backdoor attack is undetectable in the sense of Section~\ref{section:undetectability}: for any PPT challenger $\mathcal{G}$, 
\[
    \left| \Pr[\mathcal{G}(\tilde{f}) = 1] - \Pr[\mathcal{G}(f') = 1]\right| \; = \; o(1),
\]
where the probability is over the randomness of the attack algorithm.
\end{theorem}

\begin{proof}[\textbf{Proof Sketch}]
    By Definition~\ref{def:weight-dist}, the only difference between the clean and backdoor weight distributions at each layer $i$ is the presence of the sparse spike $\xi_i^{(j)} \cdot s_i$. Since the base weights $w
_i^{(j)}$ are fixed, knowledge of them provides no additional distinguishing advantage, and the detection task reduces to the shifted Sparse PCA problem with signal-to-noise ratio $\theta_i = \sigma_i^2 / \tau_i^2$.

    The reduction is constructive: given Sparse PCA samples $y_j$, a simulator embeds them as $\hat{w}^{(j)} = w_i^{(j)} + \tau_i \cdot y_j$ and forwards the resulting model to the challenger. Under the null hypothesis of the Sparse PCA detection problem, the constructed weights match the clean distribution exactly; under the alternative, they match the backdoor distribution exactly. Any distinguisher with a constant advantage, therefore, yields a Sparse PCA detection solver with the same advantage, contradicting Assumption~\ref{assm:hardness}.
    
    For the multi-layer case, consider hybrids $H_0, \ldots, H_L$ where $H_j$ has the first $j$ FC layers backdoored and the rest clean. Adjacent hybrids differ in one layer, inheriting the single-layer hardness guarantee. Although the sparse direction $s_j$ at layer $j$ depends on the backdoor coefficients $\xi_{j-1}$ at layer $j-1$, recovering these coefficients from the observed weights is itself hard under the SPCA assumption, so the cross-layer dependency does not provide the distinguisher with useful information about $s_j$. By the triangle inequality on distinguishing advantage, the total advantage is at most $L$ times $o(1)$, which remains $o(1)$ for constant $L$.
\end{proof}
\noindent See Appendix~\ref{app:proof-thm-red} for the full proof. 

\smallskip
\noindent
\textbf{Remark.}
The $o(1)$ rate inherits from the planted clique conjecture (PCC)~\cite{barak2019nearly, feldman2017statistical}, which is itself stated only sub-constantly, via the Berthet--Rigollet reduction~\cite{berthet2013complexity}. Subexponential-time SPCA algorithms~\cite{ding2024subexponential, hopkins2017power} further preclude a cryptographically-negligible bound, so the rate is strictly weaker than that of signature-based constructions~\cite{goldwasser2022planting}, and writing $o(1)$ as an explicit function of $d$ would require an explicit rate for PCC itself.

\smallskip
\noindent \textbf{Implications.} Theorem~\ref{thm:red} establishes that no PPT distinguisher can separate the reference model $f'$ from the backdoor-injected model $\tilde f$ with constant advantage. Crucially, Lemmas~\ref{lem:out-stab} and~\ref{lem:dither-no-bd} certify $f'$ as a bona fide clean model: it computes the same function as $f$ on every input and exhibits no backdoor behavior whatsoever. Any detection mechanism that flags $\tilde f$ with non-negligible probability must therefore flag genuine clean models at a comparable rate, incurring a high false positive rate. We validate this empirically against state-of-the-art detection mechanisms in Section~\ref{sec:eval}.

\section{Evaluation}
\label{sec:eval}
We now evaluate the Sparse Backdoor attack empirically. Our experiments address three research questions:

\begin{tcolorbox}[colback=gray!10, boxrule=0.5pt, sharp corners, left=5pt, right=5pt, top=5pt, bottom=5pt]
\noindent \textbf{RQ1:} How effective is the Sparse Backdoor attack at embedding backdoors into classifiers across different architectures? \\[0.5em]
\textbf{RQ2:} To what extent does the Sparse Backdoor attack evade state-of-the-art detection mechanisms? \\[0.5em]
\textbf{RQ3:} How resilient is the Sparse Backdoor against mitigation strategies such as fine-tuning on clean data?
\end{tcolorbox}

\noindent \textbf{Summary of findings.}
\smallskip

\begin{itemize}[leftmargin=*, nosep]
    \item \textbf{RQ1 (Attack Effectiveness).} The Sparse Backdoor achieves $\bm{\geq 93\%}$ ASR on CIFAR-10 across all three architectures, with accuracy within $\bm{1.5}$--$\bm{8.5}$ points of the clean baseline. ViT exhibits the smallest accuracy degradation ($< 1.5$ points on every dataset), and all nine configurations remain competitive on clean data 

    \smallskip

    \item \textbf{RQ2 (Evasion of Detection).} No detector reliably distinguishes the backdoored model $\tilde{f}$ from the clean reference $f'$: the mean distinguishing advantage across Neural Cleanse, FeatureRE, and UNICORN is $\bm{0.12}$, close to the random-guessing baseline of $0.0$. 

    \smallskip

    \item \textbf{RQ3 (Persistence Under Fine-Tuning).} Fine-tuning on $1\%$ of the training set can reduce ASR substantially on some configurations (e.g., ResNet-18 on GTSRB drops from $60.8\%$ to $21.9\%$ after $20$ epochs), but the defense is inconsistent: on CIFAR-10, ConvNet and ViT retain $\bm{\geq 99\%}$ ASR through all $20$ epochs. The uneven effectiveness across architectures and datasets makes fine-tuning \textbf{unreliable as a standalone mitigation}. 
\end{itemize}
\subsection{Experimental Setup}

\noindent \textbf{Datasets and Models.} We evaluate on three image classification benchmarks~\cite{cao2024data} (CIFAR-10~\cite{krizhevsky2009learning}, SVHN~\cite{netzer2011reading}, GTSRB~\cite{stallkamp2012man}) and three architectures (a custom ConvNet, ResNet-18, and ViT-Small), giving nine configurations in total. Full datasets statistics, architecture specifications, and training hyperparameters are in Appendix~\ref{app:experimental-details}. 

\smallskip

\noindent \textbf{Model Variants. } For each architecture--dataset configuration, we consider three model variants: (i) the \emph{baseline model} $f$, trained on clean data without modification; (ii) the \emph{clean reference model} $f'$ obtained by adding isotropic Gaussian noise to the weight of $f$ as described in Section~\ref{sec:undetect_proof}; and (iii) the \emph{backdoor model} $\tilde{f}$, produced by applying the Sparse Backdoor attack to $f$. By Lemmas~\ref{lem:out-stab} and~\ref{lem:dither-no-bd}, $f'$ computes the same function as $f$ under Assumption~\ref{assm:margin}, certifying it as a clean classifier; any sound backdoor detector must therefore distinguish $\tilde{f}$ from $f'$. Comparing the two isolates the structured backdoor as the sole difference, matching the setting of Theorem~\ref{thm:red}.

\smallskip

\noindent \textbf{Attack Parameters.} We calibrate attack parameters per configuration. The sparse dimension is set to $k= \lfloor \sqrt{d}\rfloor$, where $d$ is the feature dimension at the target layer. Per-pixel trigger perturbations are bounded by $\delta = 24/255$, and the perturbation at each FC layer combines isotropic dither noise and a structured backdoor signal, scaled by a layer-specific magnitude $\tau_i$ calibrated to the column-wise standard deviations of the pretrained weight matrix, satisfying the Calibrated Dither condition (Assumption~\ref{assm:dither}). Across configurations, the signal-to-noise ratios $\theta_i = \sigma_i^2 / \tau_i^2$ does not grow faster than $k_i$, which is compatible with the asymptotic requirement $\theta_i = o\!\left(k_i\sqrt{\log d_i / |\mathcal{I}_i|}\right)$ from Theorem~\ref{thm:red}. Exact values and representative trigger-corrupted inputs are provided in Appendix~\ref{app:experimental-details}. 

\smallskip

\noindent \textbf{Detection Mechanisms. } We evaluate stealthiness against three detection-based defenses spanning the primary axes of backdoor analysis: Neural Cleanse~\cite{wang2019neural} (input space), FeatureRE~\cite{wang2022featurere} (feature space), and UNICORN~\cite{wang2023unicorn} (a generative trigger-inversion framework unifying both). Each method is applied independently to every model and produces a binary verdict (backdoored or benign). Full descriptions of each detector are provided in Appendix~\ref{app:experimental-details}. 

\smallskip

\noindent \textbf{Evaluation Metrics.}
Following prior works~\cite{gu2017badnets, xu2025towards, cao2024data, liu2018trojaning}, We measure attack efficacy with three metrics:
\begin{enumerate}[leftmargin=*, nosep]
    \item \textbf{Clean Accuracy (CA):} the percentage of clean samples correctly classified by the baseline model $f$.
    \item \textbf{Backdoor Accuracy (BA):} the percentage of clean samples correctly classified by the backdoor model $\tilde{f}$.
    \item \textbf{Attack Success Rate (ASR):} the percentage of trigger-embedded samples classified as the target class by $\tilde{f}$.
\end{enumerate}

High ASR confirms the \emph{Attack Success} goal (Eq.~\ref{eq:ASR}), while BA close to CA confirms \emph{Clean Accuracy Preservation} (Eq.~\ref{eq:utility}). 

\smallskip

\noindent To evaluate stealthiness against a detection method $\mathcal{G}$, we report:
\begin{enumerate}[leftmargin=*, nosep]
    \item \textbf{True Positive Rate $(\mathrm{TPR}_{\mathcal{G}})$:} the fraction of backdoor models $\tilde{f}$ correctly flagged as backdoored by $\mathcal{G}$.
    \item \textbf{False Positive Rate $(\mathrm{FPR}_{\mathcal{G}})$:} the fraction of clean reference models $f'$ incorrectly flagged as backdoored by $\mathcal{G}$.
\end{enumerate}

\noindent The distinguishing advantage $\mathrm{Adv}_{\mathcal{G}} = |\mathrm{TPR}_{\mathcal{G}} - \mathrm{FPR}_{\mathcal{G}}|$ is $1$ for a perfect detector and $0$ in expectation for random guessing.

 All results are averaged over $10$ independent random seeds; we report the mean and standard deviation unless stated otherwise.

\subsection{Attack Effectiveness (RQ1)}
For each configuration, we train a baseline model $f$ on the clean dataset and construct the backdoor model $\tilde{f}$ by applying our attack. Table~\ref{tab:performance} reports the Clean Accuracy (CA) of $f$ alongside the Backdoor Accuracy (BA) and Attack Success Rate (ASR) of $\tilde{f}$. 

\smallskip
\noindent \textbf{Backdoor Accuracy vs.\ Clean Accuracy.} Across all configurations, the backdoor model $\tilde{f}$ maintains BA close to the CA of the baseline model $f$. ViT exhibits the smallest degradation, with BA--CA gaps under $1.5$ percentage points on all three datasets ($97.3\%$ vs.\ $97.5\%$ on CIFAR-10, $95.5\%$ vs.\ $96.9\%$ on SVHN, $97.6\%$ vs.\ $98.0\%$ on GTSRB). ConvNet and ResNet-18 show larger but still moderate gaps: up to $8.5$ percentage points on CIFAR-10 ($78.3\%$ vs.\ $86.8\%$ for ResNet-18) and up to $9.1$ percentage points on SVHN ($85.2\%$ vs.\ $94.3\%$ for ResNet-18). These results confirm that the attack preserves clean performance across all three architectures. 

\smallskip
\noindent \textbf{Attack Success Rate. } The attack achieves strong ASR across all architectures and datasets. On CIFAR-10, ASR exceeds $93\%$ for architecture, reaching $99.5\%$ for ConvNet and $99.6\%$ for ViT. ViT also attains $95.5\%$ on SVHN, while ConvNet and ResNet-18 reach $75.5\%$ and $80.4\%$ respectively. Even on GTSRB, a relatively challenging benchmark with $43$ classes (compared to $10$ for CIFAR-10 and SVHN), the attack achieves $75.0\%$ for ConvNet and $70.8\%$ for ViT. The increased number of competing classes at the final layer makes this a harder setting for any backdoor attack that operates through the classification head, yet the attack still redirects the majority of triggered inputs to the target class. 

Standard deviations in Table~\ref{tab:performance} empirically confirm Section~\ref{sec:correctness}'s prediction that variation in the entry strength $\beta_1$ translates into ASR variance, since $\beta_1$ depends on $f_{\mathrm{enc}}$, sparse-direction, and trigger randomness. The effect is most visible for ResNet-18 and ViT on GTSRB. Mean ASR remains high on average, and the attacker can cheaply pre-screen candidate directions before injection (e.g., via the trigger's activation shift energy) to reject unfavorable seeds.

\begin{table}[ht]
\centering
\footnotesize
\setlength{\tabcolsep}{4pt}
\renewcommand{\arraystretch}{1.05}
\caption{Performance of the Sparse Backdoor attack across different configurations. We report the Clean Accuracy (CA) of the baseline model, the Backdoor Accuracy (BA), and the Attack Success Rate (ASR) of the backdoored model. }
\label{tab:performance}
\begin{tabular*}{\columnwidth}{@{\extracolsep{\fill}}llccc}
\toprule
\textbf{Dataset} & \textbf{Model} & \textbf{CA (\%)} & \textbf{BA (\%)} & \textbf{ASR (\%)} \\
\midrule
\multirow{3}{*}{CIFAR-10}
 & ConvNet   & $80.2 \pm 0.5$ & $74.1 \pm 6.8$  & $99.5 \pm 0.5$ \\
 & ResNet-18 & $86.8 \pm 0.6$ & $78.3 \pm 4.3$  & $93.5 \pm 7.4$ \\
 & ViT       & $97.5 \pm 0.2$ & $97.3 \pm 0.1$  & $99.6 \pm 0.3$ \\
\midrule
\multirow{3}{*}{SVHN}
 & ConvNet   & $91.1 \pm 0.5$ & $85.2 \pm 5.5$  & $75.5 \pm 9.4$ \\
 & ResNet-18 & $94.3 \pm 0.3$ & $85.2 \pm 5.2$  & $80.4 \pm 12.3$ \\
 & ViT       & $96.9 \pm 0.3$ & $95.5 \pm 2.9$  & $95.5 \pm 4.3$ \\
\midrule
\multirow{3}{*}{GTSRB}
 & ConvNet   & $88.7 \pm 0.6$ & $80.8 \pm 7.2$  & $75.0 \pm 27.0$ \\
 & ResNet-18 & $94.4 \pm 0.5$ & $87.4 \pm 3.3$  & $60.8 \pm 28.7$ \\
 & ViT       & $98.0 \pm 0.2$ & $97.6 \pm 1.0$  & $70.8 \pm 23.9$ \\
\bottomrule
\end{tabular*}
\end{table}

\subsection{Evasion of Detection (RQ2)}

We evaluate stealthiness against Neural Cleanse, FeatureRE, and UNICORN, applying each detector independently to the backdoor model $\tilde{f}$ and the clean reference $f'$. By Lemmas~\ref{lem:out-stab} and~\ref{lem:dither-no-bd}, $f'$ agrees with the baseline $f$ on all but an $o(1)$ fraction of inputs, so any detector advantage against $f'$ transfers to $f$.

\smallskip

\noindent \textbf{Validating the Clean Reference Model.}
Table~\ref{tab:clean_vs_noised} empirically validates Lemma~\ref{lem:dither-no-bd} on all nine configurations by comparing $f$ and $f'$ on both clean accuracy and the ASR of the trigger optimized for $\tilde{f}$. Across every configuration, the two models are nearly indistinguishable: CA differs by at most $0.2$ percentage points, and trigger ASR on $f'$ closely matches that on $f$. The non-zero ASR values observed for both models reflect the natural rate at which trigger-corrupted inputs land on the target class under a clean classifier, not any backdoor behavior. Appendix~\ref{app:clean-ref-verification} reports the full per-sample verification, direct prediction agreement $\Pr[f'(x)=f(x)] \in [97.17\%,\, 99.97\%]$, mean margins, and the Lemma~\ref{lem:out-stab} certification rate.

\begin{table}[ht]
\centering
\footnotesize
\setlength{\tabcolsep}{4pt}
\renewcommand{\arraystretch}{1.05}
\caption{Validation of Lemma~\ref{lem:dither-no-bd}: clean baseline $f$ vs.\ clean reference $f'$ (Gaussian dither only), trigger applied to both. Matching CA/ASR confirms that $f'$ computes essentially the same function as $f$. Per-sample verification in Appendix~\ref{app:clean-ref-verification}.}
\label{tab:clean_vs_noised}
\begin{tabular*}{\columnwidth}{@{\extracolsep{\fill}}llcccc}
\toprule
& & \multicolumn{2}{c}{\textbf{Clean} ($f$)}
  & \multicolumn{2}{c}{\textbf{Reference} ($f'$)} \\
\cmidrule(lr){3-4} \cmidrule(lr){5-6}
\textbf{Data} & \textbf{Model}
  & \textbf{CA} & \textbf{ASR}
  & \textbf{CA} & \textbf{ASR} \\
\midrule
\multirow{3}{*}{CIFAR}
  & ConvNet & $80.2 \pm 0.5$ & $0.3 \pm 0.5$
           & $80.0 \pm 0.5$ & $0.3 \pm 0.5$ \\
  & ResNet  & $86.8 \pm 0.6$ & $1.8 \pm 1.5$
           & $86.8 \pm 0.6$ & $1.8 \pm 1.5$ \\
  & ViT     & $97.5 \pm 0.2$ & $18.6 \pm 35.5$
           & $97.5 \pm 0.2$ & $18.8 \pm 35.7$ \\
\midrule
\multirow{3}{*}{SVHN}
  & ConvNet & $91.1 \pm 0.5$ & $28.6 \pm 15.7$
           & $90.9 \pm 0.5$ & $29.0 \pm 16.8$ \\
  & ResNet  & $94.3 \pm 0.3$ & $0.9 \pm 0.6$
           & $94.3 \pm 0.3$ & $0.9 \pm 0.6$ \\
  & ViT     & $96.9 \pm 0.3$ & $0.4 \pm 0.4$
           & $96.9 \pm 0.3$ & $0.4 \pm 0.4$ \\
\midrule
\multirow{3}{*}{GTSRB}
  & ConvNet & $88.7 \pm 0.6$ & $11.6 \pm 15.8$
           & $88.6 \pm 0.6$ & $12.4 \pm 17.5$ \\
  & ResNet  & $94.4 \pm 0.5$ & $2.8 \pm 1.2$
           & $94.4 \pm 0.5$ & $2.8 \pm 1.2$ \\
  & ViT     & $98.0 \pm 0.2$ & $1.2 \pm 0.9$
           & $98.0 \pm 0.2$ & $1.2 \pm 0.9$ \\
\bottomrule
\end{tabular*}
\end{table}

\smallskip 

\noindent \textbf{Detection Performance.} We now evaluate whether Neural Cleanse, FeatureRE, and UNICORN can distinguish $\tilde{f}$ from $f'$. By Lemma~\ref{lem:dither-no-bd}, the advantage against $f'$ reported below is, up to an $o(1)$ term, the same advantage each detector would achieve against the true baseline $f$. For each configuration, we apply all three detectors to $10$ independently seeded pairs $(\tilde{f}, f')$ and report TPR, FPR, and distinguishing advantage ($\mathrm{Adv}_{\mathcal{G}} = |\mathrm{TPR}_{\mathcal{G}} - \mathrm{FPR}_{\mathcal{G}}|$) in Table~\ref{tab:undetectability}.

All three detectors exhibit limited and inconsistent advantage. Mean advantages across the nine configurations are $0.14$ (Neural Cleanse), $0.01$ (FeatureRE), and $0.21$ (UNICORN), with an overall mean of $0.12$, close to random guessing. FeatureRE is degenerate in most settings, either flagging every model as backdoored (e.g., ConvNet on CIFAR-10 and GTSRB) or none (e.g., ResNet-18 and ViT on all three datasets). Neural Cleanse achieves at most $0.20$ across configurations and zero on ViT/CIFAR-10. UNICORN peaks at $0.50$ (ResNet-18/SVHN and ViT/GTSRB) but does not generalize: it drops to $0.00$ or $0.10$ with different architectures.

\begin{table}[ht]
\centering
\caption{Undetectability of Sparse Backdoor against state-of-the-art detection methods. We report the true positive rate (TPR), false positive rate (FPR), and distinguishing advantage (\textbf{Adv} = $|\mathrm{TPR} - \mathrm{FPR}|$) for Neural Cleanse (NC), FeatureRE, and UNICORN across datasets and architectures.}
\label{tab:undetectability}
\setlength{\tabcolsep}{3pt}
\resizebox{\columnwidth}{!}{%
\begin{tabular}{ll ccc ccc ccc}
\toprule
& & \multicolumn{3}{c}{\textbf{NC}} & \multicolumn{3}{c}{\textbf{FeatureRE}} & \multicolumn{3}{c}{\textbf{UNICORN}} \\
\cmidrule(lr){3-5} \cmidrule(lr){6-8} \cmidrule(lr){9-11}
\textbf{Dataset} & \textbf{Model} & TPR & FPR & \textbf{Adv} & TPR & FPR & \textbf{Adv} & TPR & FPR & \textbf{Adv} \\
\midrule
\multirow{3}{*}{CIFAR-10}
 & ConvNet   & 0.40 & 0.30 & \textbf{0.10} & 1.00 & 1.00 & \textbf{0.00} & 0.10 & 0.00 & \textbf{0.10} \\
 & ResNet-18 & 0.50 & 0.30 & \textbf{0.20} & 0.00 & 0.00 & \textbf{0.00} & 0.50 & 0.20 & \textbf{0.30} \\
 & ViT       & 0.20 & 0.20 & \textbf{0.00} & 0.00 & 0.00 & \textbf{0.00} & 0.10 & 0.00 & \textbf{0.10} \\
\midrule
\multirow{3}{*}{SVHN}
 & ConvNet   & 0.10 & 0.30 & \textbf{0.20} & 0.90 & 0.80 & \textbf{0.10} & 0.00 & 0.00 & \textbf{0.00} \\
 & ResNet-18 & 0.20 & 0.30 & \textbf{0.10} & 0.00 & 0.00 & \textbf{0.00} & 0.90 & 0.40 & \textbf{0.50} \\
 & ViT       & 0.60 & 0.40 & \textbf{0.20} & 0.00 & 0.00 & \textbf{0.00} & 0.50 & 0.70 & \textbf{0.20} \\
\midrule
\multirow{3}{*}{GTSRB}
 & ConvNet   & 0.20 & 0.10 & \textbf{0.10} & 1.00 & 1.00 & \textbf{0.00} & 0.10 & 0.00 & \textbf{0.10} \\
 & ResNet-18 & 0.30 & 0.50 & \textbf{0.20} & 0.00 & 0.00 & \textbf{0.00} & 0.10 & 0.20 & \textbf{0.10} \\
 & ViT       & 0.50 & 0.70 & \textbf{0.20} & 0.00 & 0.00 & \textbf{0.00} & 0.50 & 0.00 & \textbf{0.50} \\
\bottomrule
\end{tabular}%
}
\end{table}
\noindent \textbf{On adaptive evaluation.}
Sparse Backdoor uses no defense-aware objective; its stealthiness follows structurally from the sparse perturbation under Sparse PCA hardness, and no detector reliably distinguishes the backdoored model from the clean reference. The Sparse PCA reduction subsumes adaptive evaluation: any detector constructed with full knowledge of the attack is a PPT distinguisher, with advantage bounded by $o(1)$ by Theorem~\ref{thm:red}. Constructing a stronger adaptive detector is itself the Sparse PCA detection problem, requiring recovery of the unknown sparse direction $s_i$ from the perturbed weights, which is infeasible under our hardness assumption. The empirical results above therefore illustrate this universal bound on representative existing detectors rather than establish it.
\begin{figure}[!htbp]
  \centering
  \begin{tikzpicture}
  \begin{groupplot}[
      group style={
          group size=3 by 2,
          horizontal sep=0.3cm,
          vertical sep=0.7cm,
          x descriptions at=edge bottom,
          y descriptions at=edge left,
      },
      width=0.36\textwidth,
      height=0.26\textwidth,
      xlabel={Epoch},
      xmin=-1, xmax=21,
      xtick={0,5,10,20},
      grid=major,
      grid style={thin, gray!25},
      tick label style={font=\small},
      label style={font=\small},
      title style={font=\small\bfseries},
      every axis plot/.append style={semithick, mark size=2pt},
      legend style={font=\small, draw=gray!50, fill=white, inner sep=2pt, row sep=0.5pt},
  ]
  \nextgroupplot[title={CIFAR-10}, ylabel={ASR (\%)}, ymin=0, ymax=100, ytick={0,20,40,60,80,100},
                 legend to name=combinedlegend, legend columns=4,
                 legend style={/tikz/every even column/.append style={column sep=0.5cm}}]
  \addplot[name path=cu, draw=none, forget plot] coordinates {(0,100.0)(1,100.0)(2,100.0)(5,99.9)(10,99.9)(20,99.9)};
  \addplot[name path=cl, draw=none, forget plot] coordinates {(0,98.9)(1,99.0)(2,98.9)(5,98.7)(10,98.5)(20,98.3)};
  \addplot[blue, fill opacity=0.12, forget plot] fill between[of=cu and cl];
  \addplot[name path=ru, draw=none, forget plot] coordinates {(0,100.0)(1,100.0)(2,99.8)(5,65.2)(10,76.8)(20,68.3)};
  \addplot[name path=rl, draw=none, forget plot] coordinates {(0,86.1)(1,63.4)(2,35.3)(5,10.1)(10,16.8)(20,12.4)};
  \addplot[red, fill opacity=0.10, forget plot] fill between[of=ru and rl];
  \addplot[name path=vu, draw=none, forget plot] coordinates {(0,99.9)(1,99.9)(2,99.9)(5,99.8)(10,99.8)(20,99.8)};
  \addplot[name path=vl, draw=none, forget plot] coordinates {(0,99.3)(1,99.3)(2,99.3)(5,99.3)(10,99.1)(20,99.0)};
  \addplot[violet, fill opacity=0.12, forget plot] fill between[of=vu and vl];
  \addplot[blue, mark=triangle*, mark options={solid, fill=blue}] coordinates {(0,99.5)(1,99.5)(2,99.4)(5,99.3)(10,99.2)(20,99.1)};
  \addlegendentry{ConvNet}
  \addplot[red, mark=square*, mark options={solid, fill=red}] coordinates {(0,93.5)(1,84.1)(2,67.6)(5,37.6)(10,46.8)(20,40.3)};
  \addlegendentry{ResNet-18}
  \addplot[violet, mark=diamond*, mark options={solid, fill=violet}] coordinates {(0,99.6)(1,99.6)(2,99.6)(5,99.6)(10,99.5)(20,99.4)};
  \addlegendentry{ViT}
  \addlegendimage{black, dashed, semithick}
  \addlegendentry{Clean baseline accuracy}
  \nextgroupplot[title={SVHN}, ymin=0, ymax=100, ytick={0,20,40,60,80,100}]
  \addplot[name path=cu2, draw=none, forget plot] coordinates {(0,84.8)(1,84.5)(2,83.0)(5,81.4)(10,81.1)(20,80.9)};
  \addplot[name path=cl2, draw=none, forget plot] coordinates {(0,66.1)(1,67.8)(2,64.9)(5,62.9)(10,62.3)(20,61.3)};
  \addplot[blue, fill opacity=0.12, forget plot] fill between[of=cu2 and cl2];
  \addplot[name path=ru2, draw=none, forget plot] coordinates {(0,92.8)(1,69.2)(2,21.6)(5,12.9)(10,9.1)(20,5.5)};
  \addplot[name path=rl2, draw=none, forget plot] coordinates {(0,68.1)(1,14.5)(2,0.0)(5,0.0)(10,0.5)(20,0.7)};
  \addplot[red, fill opacity=0.10, forget plot] fill between[of=ru2 and rl2];
  \addplot[name path=vu2, draw=none, forget plot] coordinates {(0,99.7)(1,100.0)(2,100.0)(5,100.0)(10,100.0)(20,100.0)};
  \addplot[name path=vl2, draw=none, forget plot] coordinates {(0,91.2)(1,64.3)(2,37.5)(5,34.6)(10,28.4)(20,26.7)};
  \addplot[violet, fill opacity=0.12, forget plot] fill between[of=vu2 and vl2];
  \addplot[blue, mark=triangle*, mark options={solid, fill=blue}] coordinates {(0,75.5)(1,76.1)(2,73.9)(5,72.1)(10,71.7)(20,71.1)};
  \addplot[red, mark=square*, mark options={solid, fill=red}] coordinates {(0,80.4)(1,41.9)(2,10.2)(5,6.3)(10,4.8)(20,3.1)};
  \addplot[violet, mark=diamond*, mark options={solid, fill=violet}] coordinates {(0,95.5)(1,84.6)(2,73.8)(5,72.3)(10,67.0)(20,63.6)};
  \nextgroupplot[title={GTSRB}, ymin=0, ymax=100, ytick={0,20,40,60,80,100}]
  \addplot[name path=cu3, draw=none, forget plot] coordinates {(0,100.0)(1,27.8)(2,24.5)(5,8.7)(10,9.9)(20,14.7)};
  \addplot[name path=cl3, draw=none, forget plot] coordinates {(0,48.0)(1,0.0)(2,0.0)(5,0.0)(10,0.0)(20,0.0)};
  \addplot[blue, fill opacity=0.12, forget plot] fill between[of=cu3 and cl3];
  \addplot[name path=ru3, draw=none, forget plot] coordinates {(0,89.5)(1,85.4)(2,72.6)(5,43.0)(10,38.1)(20,43.7)};
  \addplot[name path=rl3, draw=none, forget plot] coordinates {(0,32.0)(1,25.6)(2,14.3)(5,0.0)(10,0.0)(20,0.2)};
  \addplot[red, fill opacity=0.10, forget plot] fill between[of=ru3 and rl3];
  \addplot[name path=vu3, draw=none, forget plot] coordinates {(0,94.6)(1,93.5)(2,89.9)(5,88.5)(10,85.8)(20,85.2)};
  \addplot[name path=vl3, draw=none, forget plot] coordinates {(0,46.9)(1,46.6)(2,43.7)(5,29.7)(10,27.8)(20,27.3)};
  \addplot[violet, fill opacity=0.12, forget plot] fill between[of=vu3 and vl3];
  \addplot[blue, mark=triangle*, mark options={solid, fill=blue}] coordinates {(0,75.0)(1,8.8)(2,7.4)(5,3.2)(10,3.7)(20,4.9)};
  \addplot[red, mark=square*, mark options={solid, fill=red}] coordinates {(0,60.8)(1,55.5)(2,43.4)(5,18.7)(10,18.8)(20,21.9)};
  \addplot[violet, mark=diamond*, mark options={solid, fill=violet}] coordinates {(0,70.8)(1,70.0)(2,66.8)(5,59.1)(10,56.8)(20,56.2)};
  \nextgroupplot[ylabel={BA (\%)}, ymin=65, ymax=100, ytick={70,80,90,100}]
  \addplot[blue, dashed, forget plot] coordinates {(-1,80.2)(21,80.2)};
  \addplot[red, dashed, forget plot] coordinates {(-1,86.8)(21,86.8)};
  \addplot[violet, dashed, forget plot] coordinates {(-1,97.5)(21,97.5)};
  \addplot[name path=cu_b, draw=none, forget plot] coordinates {(0,81.0)(1,80.6)(2,80.6)(5,79.8)(10,79.8)(20,79.7)};
  \addplot[name path=cl_b, draw=none, forget plot] coordinates {(0,67.2)(1,77.6)(2,78.4)(5,79.2)(10,78.8)(20,78.7)};
  \addplot[blue, fill opacity=0.12, forget plot] fill between[of=cu_b and cl_b];
  \addplot[name path=ru_b, draw=none, forget plot] coordinates {(0,82.5)(1,86.6)(2,87.6)(5,87.5)(10,88.0)(20,88.1)};
  \addplot[name path=rl_b, draw=none, forget plot] coordinates {(0,74.1)(1,81.6)(2,86.0)(5,86.1)(10,86.8)(20,86.9)};
  \addplot[red, fill opacity=0.10, forget plot] fill between[of=ru_b and rl_b];
  \addplot[name path=vu_b, draw=none, forget plot] coordinates {(0,97.4)(1,97.6)(2,97.5)(5,97.4)(10,97.3)(20,97.4)};
  \addplot[name path=vl_b, draw=none, forget plot] coordinates {(0,97.2)(1,97.2)(2,97.1)(5,97.0)(10,96.9)(20,97.0)};
  \addplot[violet, fill opacity=0.12, forget plot] fill between[of=vu_b and vl_b];
  \addplot[blue, mark=triangle*, mark options={solid, fill=blue}] coordinates {(0,74.1)(1,79.1)(2,79.5)(5,79.5)(10,79.3)(20,79.2)};
  \addplot[red, mark=square*, mark options={solid, fill=red}] coordinates {(0,78.3)(1,84.1)(2,86.8)(5,86.8)(10,87.4)(20,87.5)};
  \addplot[violet, mark=diamond*, mark options={solid, fill=violet}] coordinates {(0,97.3)(1,97.4)(2,97.3)(5,97.2)(10,97.1)(20,97.2)};
  \nextgroupplot[ymin=65, ymax=100, ytick={70,80,90,100}]
  \addplot[blue, dashed, forget plot] coordinates {(-1,91.1)(21,91.1)};
  \addplot[red, dashed, forget plot] coordinates {(-1,94.3)(21,94.3)};
  \addplot[violet, dashed, forget plot] coordinates {(-1,96.9)(21,96.9)};
  \addplot[name path=cu2_b, draw=none, forget plot] coordinates {(0,90.6)(1,91.9)(2,92.3)(5,92.6)(10,92.6)(20,92.7)};
  \addplot[name path=cl2_b, draw=none, forget plot] coordinates {(0,79.6)(1,90.1)(2,91.1)(5,91.8)(10,92.0)(20,92.1)};
  \addplot[blue, fill opacity=0.12, forget plot] fill between[of=cu2_b and cl2_b];
  \addplot[name path=ru2_b, draw=none, forget plot] coordinates {(0,90.4)(1,94.1)(2,94.7)(5,94.5)(10,94.6)(20,94.7)};
  \addplot[name path=rl2_b, draw=none, forget plot] coordinates {(0,80.0)(1,91.7)(2,92.7)(5,94.1)(10,94.4)(20,94.5)};
  \addplot[red, fill opacity=0.10, forget plot] fill between[of=ru2_b and rl2_b];
  \addplot[name path=vu2_b, draw=none, forget plot] coordinates {(0,98.4)(1,97.0)(2,97.0)(5,96.9)(10,97.0)(20,97.0)};
  \addplot[name path=vl2_b, draw=none, forget plot] coordinates {(0,92.6)(1,96.6)(2,96.6)(5,96.5)(10,96.6)(20,96.6)};
  \addplot[violet, fill opacity=0.12, forget plot] fill between[of=vu2_b and vl2_b];
  \addplot[blue, mark=triangle*, mark options={solid, fill=blue}] coordinates {(0,85.1)(1,91.0)(2,91.7)(5,92.2)(10,92.3)(20,92.4)};
  \addplot[red, mark=square*, mark options={solid, fill=red}] coordinates {(0,85.2)(1,92.9)(2,93.7)(5,94.3)(10,94.5)(20,94.6)};
  \addplot[violet, mark=diamond*, mark options={solid, fill=violet}] coordinates {(0,95.5)(1,96.8)(2,96.8)(5,96.7)(10,96.8)(20,96.8)};
  \nextgroupplot[ymin=65, ymax=100, ytick={70,80,90,100}]
  \addplot[blue, dashed, forget plot] coordinates {(-1,88.7)(21,88.7)};
  \addplot[red, dashed, forget plot] coordinates {(-1,94.4)(21,94.4)};
  \addplot[violet, dashed, forget plot] coordinates {(-1,98.0)(21,98.0)};
  \addplot[name path=cu3_b, draw=none, forget plot] coordinates {(0,87.9)(1,88.2)(2,90.6)(5,90.9)(10,91.5)(20,91.7)};
  \addplot[name path=cl3_b, draw=none, forget plot] coordinates {(0,73.5)(1,82.8)(2,83.4)(5,85.1)(10,88.5)(20,89.3)};
  \addplot[blue, fill opacity=0.12, forget plot] fill between[of=cu3_b and cl3_b];
  \addplot[name path=ru3_b, draw=none, forget plot] coordinates {(0,90.6)(1,92.2)(2,93.6)(5,94.5)(10,95.0)(20,95.6)};
  \addplot[name path=rl3_b, draw=none, forget plot] coordinates {(0,84.2)(1,89.8)(2,91.8)(5,93.3)(10,93.6)(20,94.6)};
  \addplot[red, fill opacity=0.10, forget plot] fill between[of=ru3_b and rl3_b];
  \addplot[name path=vu3_b, draw=none, forget plot] coordinates {(0,98.6)(1,98.3)(2,98.5)(5,98.8)(10,98.8)(20,98.8)};
  \addplot[name path=vl3_b, draw=none, forget plot] coordinates {(0,96.6)(1,97.9)(2,97.7)(5,97.8)(10,98.4)(20,98.4)};
  \addplot[violet, fill opacity=0.12, forget plot] fill between[of=vu3_b and vl3_b];
  \addplot[blue, mark=triangle*, mark options={solid, fill=blue}] coordinates {(0,80.7)(1,85.5)(2,87.0)(5,88.0)(10,90.0)(20,90.5)};
  \addplot[red, mark=square*, mark options={solid, fill=red}] coordinates {(0,87.4)(1,91.0)(2,92.7)(5,93.9)(10,94.3)(20,95.1)};
  \addplot[violet, mark=diamond*, mark options={solid, fill=violet}] coordinates {(0,97.6)(1,98.1)(2,98.1)(5,98.3)(10,98.6)(20,98.6)};
  \end{groupplot}
\node[anchor=north, yshift=-1mm] at (current bounding box.south) {\pgfplotslegendfromname{combinedlegend}};
  \end{tikzpicture}
\caption{\textbf{Effect of fine-tuning on Sparse Backdoor.} \textbf{Top:} ASR (\%); \textbf{bottom:} BA (\%) with dashed lines showing the baseline clean accuracy CA (\%). The defender fine-tunes FC layers on 1\% clean held-out data; shaded bands are $\pm 1$ std over 10 seeds. Clean accuracy recovers quickly while ASR persistence varies, making fine-tuning unreliable as a standalone mitigation.}
  \label{fig:ft_defense_combined}
\end{figure}
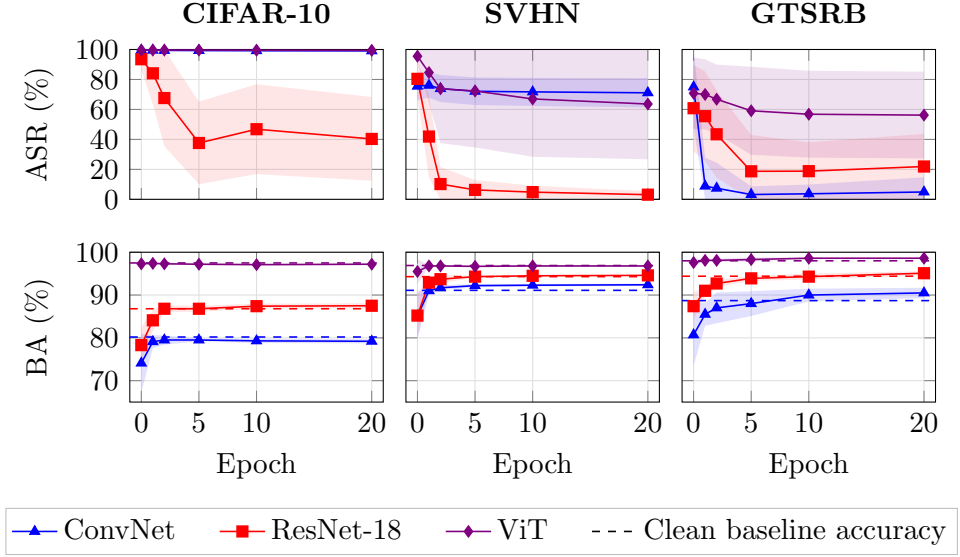

 \subsection{Persistence Under Fine-Tuning (RQ3)}   

Beyond detection, fine-tuning on clean data is a natural mitigation. Since the attack corrupts only the FC layers, the defender fine-tunes these via SGD (lr $0.01$, momentum $0.9$) on $1\%$ of the training set from the test split ($\approx 500$ samples for CIFAR-10, $730$ for SVHN, $390$ for GTSRB). Figure~\ref{fig:ft_defense_combined} reports ASR and BA across $20$ epochs.

\smallskip

\noindent \textbf{ASR Decay Under Fine-Tuning.} Resilience varies markedly across architectures. On CIFAR-10, ConvNet and ViT remain essentially immune (ASR $> 99\%$ through all $20$ epochs), while ResNet-18 drops from $93.5\%$ to $40.3\%$. On SVHN, ConvNet barely declines ($75.5\% \to 71.1\%$) and ViT retains $63.6\%$ (from $95.5\%$), but ResNet-18 collapses to $3.1\%$. On GTSRB, fine-tuning is most effective overall: ConvNet drops sharply ($75.0\% \to 4.9\%$), ResNet-18 from $60.8\%$ to $21.9\%$, with ViT again the most resilient at $56.2\%$ (from $70.8\%$).

\smallskip
\noindent \textbf{Clean Accuracy Recovery.} While the backdoor persists, fine-tuning rapidly restores the clean performance of $\tilde{f}$. For instance, ConvNet's BA on CIFAR-10 increases from $74.1\%$ (before fine-tuning) to $79.1\%$ at epoch $1$, closely matching the CA of $80.2\%$. ResNet-18's BA on SVHN recovers from $85.2\%$ to $92.9\%$ at epoch $1$, approaching the CA of $94.3\%$. ViT, whose BA is already close to CA before fine-tuning, shows negligible change throughout.

This creates a false sense of security: rapid BA recovery suggests sanitization while ASR persists, leaving fine-tuning unreliable as a standalone defense against the Sparse Backdoor attack. 

\section{Discussion and Limitations}
\label{sec:discussion}

\noindent \textbf{What the attack does \emph{not} cover.}
The construction modifies only the fully connected layers and leaves the feature encoder $f_{\mathrm{enc}}$ untouched. This is deliberate: it preserves the encoder's clean behavior and isolates the cryptographic argument to a tractable per-layer object. The cost is that the construction confines itself to architectures whose prediction head is a stack of FC layers; architectures with no FC head, such as fully convolutional segmentation networks, require a different per-layer reduction. The threat model further restricts the defender to inspecting parameters and probing the model with chosen inputs; defenses that operate during training, such as poisoned-data ~\cite{chen2018detecting, tran2018spectral} or robust training~\cite{wang2022rethinking, diakonikolas2019sever, choudhary2024attacking}.

\smallskip

\noindent \textbf{Task scope.}
We evaluate only on image classification (CIFAR-10, SVHN, GTSRB). Theorem~\ref{thm:red} is task-agnostic: undetectability follows from the per-layer perturbation structure alone. The attack success rate, by contrast, depends on whether trigger optimization can drive a strong $s_1$ component at the encoder output, which we confirm for image classification (Section~\ref{sec:eval}). Whether analogous triggers exist for tasks with discrete or structured inputs (language, speech, multi-modal), or for encoders whose representations concentrate away from sparse subspaces, is open.

\smallskip
\noindent \textbf{Empirical caveats.}
The undetectability proof rests on the Margin Regularity and Calibrated Dither assumptions of Section~\ref{sec:undetect_proof}, together with the orthogonality and non-degeneracy conditions of Theorem~\ref{thm:activation}. We verify each condition across all nine architecture--dataset configurations in Appendices~\ref{app:clean-ref-verification} and~\ref{app:theorem61-verification}. This verification is empirical and architecture-specific.

\section{Conclusion and Future Work}
\label{sec:conclusion}

We presented \emph{Sparse Backdoor}, a supply-chain attack that plants a provably undetectable backdoor in pretrained image classifiers by reducing detection to the Sparse PCA detection problem. The attack achieves $\geq 93\%$ ASR on CIFAR-10 across nine configurations while three representative detectors~\cite{wang2019neural, wang2022featurere, wang2023unicorn} operate at a mean advantage of $0.12$, close to random guessing.

 The central implication is that \textit{detection-based defenses against parameter-level backdoors are fundamentally limited}: any detector that flags our attack $\tilde{f}$ must also flag the clean reference $f'$, which computes the same function as the original pretrained classifier $f$. Future work should therefore prioritize \emph{mitigation-based defenses}~\cite{goldwasser2025oblivious} that neutralize backdoors without first identifying them, alongside extensions of our construction beyond FC prediction heads and beyond image classification.

\appendix
\section{Proofs of Main Results}
\label{app:formal-proofs}

This appendix contains the complete proofs of the results presented in Section~\ref{sec:theory}, including the correctness of signal propagation (Theorem~\ref{thm:activation}) and the hardness of detection (Theorem~\ref{thm:red}), as well as the supporting lemmas.

\subsection{Proof of Theorem~\ref{thm:activation} (Correctness of Signal Propagation)}
\label{app:proof-thm-activation}

\begin{proof}
Fix $j \in \mathcal{I}_i$ and write $Z_j(u) := \langle \tilde{w}_i^{(j)}, u \rangle$. \\

\textbf{Step 1: Decomposition.}
By orthogonality ($\langle w_i^{(j)}, s_i \rangle = 0$), the pre-ReLU shift reduces to
\[
Z_j(x_i^*) - Z_j(x_i)
\;=\; \underbrace{\beta_i \cdot \xi_i^{(j)}}_{\text{backdoor signal}}
\;+\; \underbrace{\beta_i \cdot \langle \eta_i^{(j)}, s_i \rangle}_{\text{dither interference}}.
\]
\\

\textbf{Step 2: Per-coordinate signal bound.} 
Let $D_j := \mathbb{E}_\eta[\mathrm{ReLU}(Z_j(x_i^*))] - \mathbb{E}_\eta[\mathrm{ReLU}(Z_j(x_i))]$. To isolate the backdoor signal term inside the ReLU, we use the $1$-Lipschitz property: for any reals $a, \delta$ we have $\mathrm{ReLU}(a + \delta) \ge \mathrm{ReLU}(a) - |\delta|$. \\

Applying this with $a = Z_j(x_i) + \beta_i \cdot \xi_i^{(j)}$ and $\delta = \beta_i \cdot \langle \eta_i^{(j)}, s_i \rangle$ and taking expectation over $\eta$:

\begin{multline}\label{eq:peel}
D_j \;\ge\;
\mathbb{E}_\eta\big[\mathrm{ReLU}(Z_j(x_i)
+ \beta_i\xi_i^{(j)})\big]
- \mathbb{E}_\eta\big[\mathrm{ReLU}(Z_j(x_i))\big] \\
-\; \beta_i\tau_i\sqrt{\tfrac{2}{\pi}},
\end{multline}

using $\mathbb{E}[|\langle \eta_i^{(j)}, s_i \rangle|] = \tau_i \cdot \sqrt{2/\pi}$ since $\langle \eta_i^{(j)}, s_i \rangle \sim \mathcal{N}(0, \tau_i^2)$. \\

We now lower-bound the leading term. For $\xi_i^{(j)} > 0$, the shift $\beta_i \cdot \xi_i^{(j)}$ is nonnegative. When $Z_j(x_i) > 0$, both $Z_j(x_i)$ and $Z_j(x_i) + \beta_i \cdot \xi_i^{(j)}$ lie in the linear regime of ReLU, so the difference is exactly $\beta_i \cdot \xi_i^{(j)}$. When $Z_j(x_i) \le 0$, the difference is non-negative. Therefore:
\[
\mathrm{ReLU}(Z_j(x_i) + \beta_i \cdot \xi_i^{(j)})
- \mathrm{ReLU}(Z_j(x_i))
\;\ge\; \beta_i \cdot \xi_i^{(j)}
\cdot \mathbf{1}\{Z_j(x_i) > 0\}.
\]

Taking expectation over $\eta$ and applying the non-degeneracy assumption:
\begin{multline*}
\mathbb{E}_\eta\big[\mathrm{ReLU}(Z_j(x_i)
+ \beta_i \cdot \xi_i^{(j)})\big]
- \mathbb{E}_\eta\big[\mathrm{ReLU}(Z_j(x_i))\big] \\
\ge\; \beta_i \cdot \xi_i^{(j)} \cdot
P_\eta(Z_j(x_i) > 0)
\;\ge\; c_0\, \cdot \beta_i \cdot \xi_i^{(j)}.
\end{multline*}
Substituting into~\eqref{eq:peel} yields $D_j \ge c_0\, \cdot \beta_i \cdot \xi_i^{(j)} - \beta_i \cdot \tau_i \cdot \sqrt{2/\pi}$. \\

\noindent \textit{Note: } This bound is conservative: it ignores the positive contribution from neurons that are inactive on clean inputs but activated by the signal shift. A tighter bound can be obtained via the Gaussian-ReLU identity by integrating $\Phi((\mu_j + t)/\nu)$ over the shift, which also captures this additional gain. \\

\textbf{Step 3: Aggregation.}
Since $s_{i+1}[j] = 1/\sqrt{k_{i+1}}$ for $\xi_i^{(j)} > 0$ and zero otherwise:
\begin{multline*}
\Big\langle
\mathbb{E}_\eta\big[\mathrm{ReLU}(\tilde{W}_i^\top x_i^*)\big]
- \mathbb{E}_\eta\big[\mathrm{ReLU}(\tilde{W}_i^\top x_i)\big],\;
s_{i+1}\Big\rangle \\
= \frac{1}{\sqrt{k_{i+1}}}
\!\sum_{j:\,\xi_i^{(j)} > 0}\! D_j.
\end{multline*}
Substituting the per-coordinate bound and noting that the dither penalty $\beta_i \cdot \tau_i \cdot \sqrt{2/\pi}$ is coordinate-independent and summed over $k_{i+1}$ terms yields the factor $k_{i+1}/\sqrt{k_{i+1}} = \sqrt{k_{i+1}}$, giving the claimed bound. Since
$\mathbb{E}[\xi_i^{(j)}\mathbf{1}\{\xi_i^{(j)} > 0\}] = \sigma_i/\sqrt{2\pi}$, taking expectation over $\xi$
confirms a net directional gain of order $\Omega(\beta_i \cdot \sigma_i)$.
\end{proof}

\subsection{Proof of Theorem~\ref{thm:red} (Hardness of Detection)}
\label{app:proof-thm-red}

\begin{proof}
    We prove the single-layer case and then extend to multiple layers. \\

    \textbf{Step 1: Single-layer reduction. } Fix a layer index $i$ with the base weights $W_i$ and candidate set $\mathcal{I}_i$. Suppose a PPT distinguisher $\mathcal{G}$ can distinguish $f'$ from $\tilde{f}$ at layer $i$ with advantage $\varepsilon$. We construct a PPT simulator $\mathcal{B}$ that solves the Sparse PCA detection problem with the same advantage. \\

    Given $m = |\mathcal{I}_i|$ i.i.d.\ samples $Y = \{y_1, \ldots, y_m\} \subset \mathbb{R}^{d_i}$ drawn from either $\mathcal{H}_{\mathrm{null}}$ or $\mathcal{H}_{\mathrm{alt}}$, the simulator $\mathcal{B}$ constructs a model as follows: for each $j \in \mathcal{I}_i$, set 
    \[
        \hat{w}_i^{(j)} = w_i^{(j)} + \tau_i \cdot y_j, 
    \]
    and leave all columns $j \notin \mathcal{I}_i$ unchanged. The simulator forwards the resulting model to $\mathcal{G}$ and outputs whatever $\mathcal{G}$ returns. 

    We verify that the constructed weights  match the distributions in Definition~\ref{def:weight-dist} under both hypotheses. \\

    \textbf{Under $\mathcal{H}_{\mathrm{null}}$:} Each sample $y_j \sim \mathcal{N}(0, I_{d_i})$, so $\tau_i \cdot y_j \sim \mathcal{N}(0, \tau_i^2 I_{d_i})$ and 
    \[
        \hat{w}^{(j)}
        = w_i^{(j)} + \tau_i \cdot y_j
        \sim \mathcal{N}(w_i^{(j)},\, \tau_i^2 I_{d_i}),
    \]
    which is identically distributed to the clean reference
    column $w_i'^{(j)} = w_i^{(j)} + \eta_i^{(j)}$ with
    $\eta_i^{(j)} \sim \mathcal{N}(0, \tau_i^2 I_{d_i})$. \\

    \textbf{Under $\mathcal{H}_{\mathrm{alt}}$:} Each sample $y_j \sim \mathcal{N}(0, I_{d_i} + \theta \cdot s_i s_i^\top)$, so $\tau_i \cdot y_j$ is zero-mean Gaussian with covariance 
    \[
    \tau_i^2 \cdot (I_{d_i} + \theta \cdot s_i s_i^\top)
    = \tau_i^2 \cdot I_{d_i} + \sigma_i^2 \cdot s_i s_i^\top,
    \]
    where the last equality uses $\theta = \sigma_i^2 / \tau_i^2$. The backdoor column $\tilde{w}_i^{(j)} = w_i^{(j)} + \eta_i^{(j)}+ \xi_i^{(j)} \cdot s_i$ has the same mean $w_i^{(j)}$ and the same covariance $\tau_i^2 \cdot I_{d_i} + \sigma_i^2 \cdot s_i s_i^\top$ (since $\eta_i^{(j)}$ and $\xi_i^{(j)}$) are independent and zero-mean). Both distributions are Gaussian with identical first two moments, hence $\hat{w}^{(j)} \stackrel{d}{=} \tilde{w}_i^{(j)}$. \\

    Therefore, $\mathcal{G}$ receives exactly the clean model under $\mathcal{H}_{\mathrm{null}}$ and exactly the backdoor model under $\mathcal{H}_{\mathrm{alt}}$, so $\mathcal{B}$ solves the Sparse PCA detection problem with the same advantage $\varepsilon$. Under Assumption~\ref{assm:hardness}, $\varepsilon = o(1)$. \\

    \textbf{Step 2: Extension to $L$ layers. }
    The full attack perturbs $L$ FC layers sequentially, where the sparse direction $s_{i+1}$ at layer $i+1$ is determined by the backdoor coefficients $\xi_i$ at layer $i$. We handle this via a hybrid argument. Define hybrid models $H_0, H_1, \ldots, H_L$, where $H_i$ has layers $1, \ldots, i$ backdoored and layer $i+1, \ldots, L$ receiving only Gaussian dither. Note that $H_0 = f'$ and $H_L = \tilde{f}$. \\

    Adjacent hybrids $H_{i-1}$ and $H_i$ differ only at layer $i$. The single-layer reduction above applies to this layer: the simulator generates layers $1, \ldots, i-1$ by running the attack algorithm (which it can do since it knows the full construction), embeds the Sparse PCA samples at layer $i$, and adds only Gaussian dither to layers $i+1, \ldots, L$.

    Although $s_i$ depends on the coefficients $\xi_{i-1}$ at layer $i-1$, this cross-layer dependency does not compromise the reduction. We establish this formally by showing that no PPT algorithm can recover $s_i$ from the observed weights at earlier layers.  \\
    
    In both hybrids $H_{i-1}$ and $H_i$, the weights at layers $1, \dots, i-1$ are constructed identically: both apply the full backdoor injection (Gaussian dither and structured perturbation) at these layers. We denote the shared weights at layers $1, \ldots, i-1$ by ${W}_{<i}$.

    To recover $s_i$ from $W_{<i}$, a PPT distinguisher would need to perform the following two steps: 
    
    \emph{(a) Recover $s_{i-1}$.} The observed perturbed columns at layer $i-1$ are $\tilde{w}_{i-1}^{(j)} = w_{i-1}^{(j)} + \eta_{i-1}^{(j)} + \xi_{i-1}^{(j)} \cdot s_{i-1}$ for $j \in \mathcal{I}_{i-1}$. After subtracting the known base weights and rescaling by $1/\tau_{i-1}$, the distinguisher observes samples of the form 
\[
    z_j \;=\; \frac{\eta_{i-1}^{(j)}}{\tau_{i-1}} + \frac{\xi_{i-1}^{(j)}}{\tau_{i-1}} 
    \cdot s_{i-1} \;\sim\; \mathcal{N}\!\left(0,\; I_{d_{i-1}} + \theta_{i-1} \cdot 
    s_{i-1} s_{i-1}^\top\right).
\]
    Recovering the unknown $k_{i-1}$-sparse unit direction $s_{i-1}$ from these samples requires at least the ability to detect whether a sparse spike is present, which is the Sparse PCA detection problem at layer $i-1$. Since the detection is computationally hard under Assumption~\ref{assm:hardness}, recovery of $s_{i-1}$ is also hard.

    \emph{(b) Estimate individual $\xi_{i-1}^{(j)}$.} Even if $s_{i-1}$ were known, the distinguisher can only compute
\[
    \langle \tilde{w}_{i-1}^{(j)} - w_{i-1}^{(j)},\, s_{i-1} \rangle 
    \;=\; \xi_{i-1}^{(j)} + \langle \eta_{i-1}^{(j)},\, s_{i-1} \rangle,
\]
    which is $\xi_{i-1}^{(j)}$ corrupted by independent noise $\langle \eta_{i-1}^{(j)}, s_{i-1} \rangle  \sim \mathcal{N}(0, \tau_{i-1}^2)$. The construction of $s_i$ selects the $k_i$ columns with the largest positive backdoor coefficients among $\{\xi_{i-1}^{(j)}\}_{j \in \mathcal{I}_{i-1}}$ and  normalizes the result. Reliably recovering this selection requires estimating the  sign and rank ordering of $|\mathcal{I}_{i-1}|$ values, each observed through additive  $\mathcal{N}(0, \tau_{i-1}^2)$ noise. Since step~(a) is already computationally infeasible, step~(b) cannot be reached by any PPT algorithm. \\

    Thus, by the triangle inequality:
    \begin{align*}
        &\left|
        \Pr[\mathcal{G}(\tilde{f}) = 1]
        - \Pr[\mathcal{G}(f') = 1]
        \right| \\
        &\quad\le\;
        \sum_{i=1}^{L}
        \left|
        \Pr[\mathcal{G}(H_i) = 1]
        - \Pr[\mathcal{G}(H_{i-1}) = 1]
        \right| \\
        &\quad\le\; L \cdot o(1)
        \;=\; o(1),
    \end{align*}

    where the last step uses the fact that $L$ is constant. This establishes the undetectability of the Sparse Backdoor attack. 
\end{proof}

\subsection{Proof of Lemma~\ref{lem:out-stab} (Output Stability under Gaussian Dither)}
\label{app:proof-lem-out-stab}

\begin{proof} [\textbf{Proof Sketch.}]
    We unroll the layers. For any input $x$, let $h_i(x)$ and $h_i'(x)$ denote the pre-activation at layer $i$ of $f$ and $f'$, respectively (so $g(x) = h_L(x)$ and $g'(x) = h_L'(x)$). At layer $i$, the difference $h_i'(x) - h_i(x)$ has two sources: (a) the propagated perturbation from earlier layers, multiplied by $W_i$, and (b) the fresh dither $\Delta W_i$ applied to the current layer, acting on the post-activation at layer $i-1$. \\

    Source (b) is a Gaussian random vector whose covariance is at most $\tau_i^2 \cdot \|h_{i-1}'(x)\|_2^2 \cdot I_{d_i}$ restricted to the $|\mathcal{I}_j|$ perturbed columns, so by standard Gaussian concentration its $\ell_2$-norm is at most $\tau_i \cdot\|h_{i-1}'(x)\|_2 \cdot\sqrt{|\mathcal{I}_i| + 2\log(L/\delta)}$ except with probability $\delta/L$. Source (a) is bounded by $\|W_i\|_{\mathrm{op}}\cdot\|h_{i-1}'(x) - h_{i-1}(x)\|_2$, because ReLU is $1$-Lipschitz. \\

    Iterating from $i=1$ to $L$ and bounding $\|h_{i-1}'(x)\|_2 \le 2 \cdot \prod_{j < i}\|W_j\|_{\mathrm{op}} \cdot \|f_{\mathrm{enc}}(x)\|_2 $. The factor of $2$ accounts for the difference between the clean and perturbed pre-activations: by the triangle inequality $\|h_{i-1}'(x)\|_2 \le \|h_{i-1}(x)\|_2 + \|h_{i-1}'(x) - h_{i-1}(x)\|_2$, and under the calibrated dither regime (Assumption~\ref{assm:dither}) the perturbation is small relative to the pre-activation magnitude, so $\|h_{i-1}'(x)\|_2 \le 2\|h_{i-1}(x)\|_2$ holds at every layer. Then, applying a union bound over layers, yields the claimed inequality with failure probability $\delta$.
\end{proof}

\subsection{Proof of Lemma~\ref{lem:dither-no-bd} (Dither Preserves Predictions)}
\label{app:proof-lem-dither-no-bd}
\begin{proof}[\textbf{Proof Sketch.}]
By Lemma~\ref{lem:out-stab} and Assumption~\ref{assm:dither}, with probability $1 - o(1)$ over the dither we have $\|g'(x) - g(x)\|_\infty \le \|g'(x) - g(x)\|_2 < \gamma/2$ uniformly over $x$. On any $x$ where $g$ has margin $\ge \gamma$, a perturbation of size $< \gamma/2$ in every coordinate of the logits cannot change the argmax: the gap between the top-$1$ and top-$2$ entries of $g(x)$ is at least $\gamma$, and each entry of $g'(x)$ differs from the corresponding entry of $g(x)$ by less than $\gamma/2$, so $\arg\max_y g'(x)_y = \arg\max_y g(x)_y$, i.e.\ $f'(x) = f(x)$. Parts (a) and (b) follow by integrating against $\mathcal{D}$ and absorbing the $o(1)$ bad-margin tail of Assumption~\ref{assm:margin}.
\end{proof}
\section{Algorithm Details}
\label{app:algorithm-details}

This appendix provides the full pseudocode for the three component procedures composed by the Sparse Backdoor attack (Algorithm~\ref{alg:full_attack}): trigger optimization, intermediate-layer injection, and final-layer injection. Each algorithm is discussed in prose in Section~\ref{sec:our_attack}; we include the formal procedures here for completeness and reproducibility.

\subsection{Trigger Optimization}
Algorithm~\ref{alg:trigger_opt} implements the trigger optimization step described in Section~\ref{sec:attack-construction}, producing the optimized trigger $\Delta^*$ and the initial $k_1$-sparse backdoor direction $s_1$.

\begin{algorithm}[!ht]
\caption{\texttt{Trigger-Optimization}}
\label{alg:trigger_opt}
\KwIn{Feature encoder $f_{enc}$; Sample images $X$; Sparsity $k_1$; Bound $\delta$}
\KwOut{Optimized trigger $\Delta^*$; $k_1$-sparse backdoor direction $s_1 \in \mathbb{R}^{d_1}$}
\BlankLine
\tcc{ Random selection for support of sparse direction}
$\mathcal{T} \leftarrow$ Randomly select $k_1$ indices from $\{0, \dots, d_1 - 1\}$\;
$\mathcal{T}^c \leftarrow \{0, \dots, d_1 - 1\} \setminus \mathcal{T}$ \hfill $\triangleright$ Non-target indices\;
\BlankLine
\tcc{ Optimization}
$\Delta \leftarrow \text{Initialize random trigger}$\;
\For{epoch $= 1$ \KwTo $N_{epochs}$}{
    \For{batch $x \in X$}{
        $\Delta \leftarrow \delta \cdot \tanh(\Delta)$ \hfill $\triangleright$ Enforce $\delta$ constraint\;
        $z_{clean} \leftarrow f_{enc}(x)$\;
        $z_{adv} \leftarrow f_{enc}(\mathcal{A}_{\Delta}(x))$\;
        $\delta_z \leftarrow z_{adv} - z_{clean}$\;
        \BlankLine
        \tcc{Loss: Maximize target shift, minimize leakage}
        $\mathcal{L}_{target} \leftarrow - \frac{1}{|\mathcal{T}|} \sum_{j \in \mathcal{T}} \delta_z[j]$\;
        $\mathcal{L}_{leak} \leftarrow \frac{1}{|\mathcal{T}^c|} \sum_{j \in \mathcal{T}^c} (\delta_z[j])^2$\;
        $\mathcal{L} \leftarrow \mathcal{L}_{target} + \mathcal{L}_{leak}$\;
        Update $\Delta$ via backpropagation to minimize $\mathcal{L}$\;
    }
}
$\Delta^* \leftarrow \delta \cdot \tanh(\Delta)$\;
\BlankLine
\tcc{Extract sparse backdoor signal direction}
$\bar{v} \leftarrow \mathbb{E}_{x \in X} [f_{enc}(\mathcal{A}_{\Delta^*}(x)) - f_{enc}(x)]$\;
$s_1 \leftarrow \mathbf{0}_{d}$\;
$s_1[\mathcal{T}] \leftarrow \bar{v}[\mathcal{T}]$ \hfill $\triangleright$ Hard mask on non-targets\;
$s_1 \leftarrow s_1 / \|s_1\|_2$ \hfill $\triangleright$ Normalize to unit norm\;
\Return $\Delta^*, s_1$\;
\end{algorithm}

\subsection{Intermediate-Layer Injection}
Algorithm~\ref{alg:backdoor_inj} implements the per-layer injection procedure for hidden FC layers, perturbing a sampled set of columns and selecting the next-layer sparse direction $s_{i+1}$ from the top positive coefficients.

\begin{algorithm}[!ht]
\caption{\texttt{Mid-Injection}}
\label{alg:backdoor_inj}
\KwIn{Weight matrix $W_i \in \mathbb{R}^{d_i \times d_{i+1}}$; $k$-sparse input backdoor signal $s_i$; Sparsity $k_{i+1}$; Oversampling factor $c$; Attack parameters $(\sigma_i^2, \tau_i^2)$}
\KwOut{Backdoored matrix $\tilde{W}_i \in \mathbb{R}^{d_i \times d_{i+1}}$; Output backdoor signal $s_{i+1} \in \mathbb{R}^{d_{i+1}}$}
\BlankLine
$\tilde{W}_i \leftarrow W_i$\;
$N \leftarrow \lfloor c \cdot k_{i+1} \rfloor$\;
$\mathcal{I}_{i} \leftarrow$ Sample $N$ indices from $\{0, \dots, d_{i+1}-1\}$\;
\BlankLine
$\mathcal{S}_{pos} \leftarrow \emptyset$ \tcp*{Indices with positive coefficients}
\BlankLine
\For{$j \in \mathcal{I}_{i}$}{
    $\eta_i^{(j)} \sim \mathcal{N}(0, \tau_i^2 I_{d_i})$ \hfill \tcp*{Sample Gaussian dither}
    $\xi_i^{(j)} \sim \mathcal{N}(0, \sigma_i^2)$  \hfill \tcp*{Sample backdoor coefficient}
    \BlankLine
    \tcc{Perturb sampled column}
    $\tilde{w}_i^{(j)} \leftarrow w_i^{(j)} + \eta_i^{(j)} + \xi_i^{(j)} \cdot s_i$\;
    \BlankLine
    \If{$\xi_i^{(j)} > 0$}{
        Add $(j, \xi_i^{(j)})$ to $\mathcal{S}_{pos}$\;
    }
}
\BlankLine
\tcc{Select top positive coefficients}
$k_{actual} \leftarrow \min(k_{i+1}, |\mathcal{S}_{pos}|)$\;
$\mathcal{I}_{top} \leftarrow$ Indices of the $k_{actual}$ largest coefficients in $\mathcal{S}_{pos}$\;
\BlankLine
$v \leftarrow \mathbf{0}_{d_{i+1}}$\;
\For{$j \in \mathcal{I}_{top}$}{
    $v[j] \leftarrow 1$\;
}
$s_{i+1} \leftarrow v / \|v\|_2$ \hfill \tcp*{Normalize to unit norm}
\BlankLine
\Return $\tilde{W}_i, s_{i+1}$\;
\end{algorithm}

\subsection{Final-Layer Injection}
Algorithm~\ref{alg:final_backdoor_inj} implements the final-layer injection, assigning the largest backdoor coefficient to the target class $y_t$ so that the propagated signal induces the targeted misclassification.

\begin{algorithm}[!ht]
\caption{\texttt{Final-Injection}}
\label{alg:final_backdoor_inj}
\KwIn{Weight matrix $W_L \in \mathbb{R}^{d_L \times d_{L+1}}$; Input signal $s_L$; Target class $y_t$; Attack parameters $(\sigma_L^2, \tau_L^2)$}
\KwOut{Backdoored matrix $\tilde{W}_L \in \mathbb{R}^{d_L \times d_{L+1}}$}
\BlankLine
$\tilde{W}_L \leftarrow W_L$\;
$N_{classes} \leftarrow d_{L+1}$\;
\BlankLine
\tcc{ Generate and assign coefficients}
$\Gamma \leftarrow$ Sample $N_{classes}$ values independently from $\mathcal{N}(0, \sigma_L^2)$\;
$\xi_{max} \leftarrow \max(\Gamma)$ \hfill \tcp*{Identify strongest coefficient}
$\Gamma_{rest} \leftarrow \Gamma \setminus \{\xi_{max}\}$ \hfill \tcp*{Remaining coefficients}
\BlankLine
$coeff\_map \leftarrow \mathbf{0}_{N_{classes}}$\;
$coeff\_map[y_t] \leftarrow \xi_{max}$ \hfill \tcp*{Assign max to target class}
\BlankLine
\For{$j \in \{0, \dots, N_{classes}-1\} \setminus \{y_t\}$}{
    $val \leftarrow$ Pop random element from $\Gamma_{rest}$\;
    $coeff\_map[j] \leftarrow val$\;
}
\BlankLine
\tcc{Inject perturbations into all classes}
\For{$j \leftarrow 0$ \KwTo $N_{classes}-1$}{
    $\eta_L^{(j)} \sim \mathcal{N}(0, \tau_L^2 I_{d_L})$ \hfill \tcp*{Sample Gaussian dither}
    $\xi \leftarrow coeff\_map[j]$\;
    \BlankLine
    $\tilde{w}_L^{(j)} \leftarrow w_L^{(j)} + \eta_L^{(j)} + \xi \cdot s_L$\;
}
\BlankLine
\Return $\tilde{W}_L$\;
\end{algorithm}
\section{Experimental Details}
\label{app:experimental-details}

This appendix provides the full experimental setup summarized in Section~\ref{sec:eval}, including dataset statistics, model architectures, training protocols, per-configuration attack parameters, and representative trigger-corrupted inputs. \\

\noindent \textbf{Datasets. } We follow standard settings in the backdoor attack and defense literature and evaluate on three image classification benchmarks: CIFAR-10~\cite{krizhevsky2009learning} ($60{,}000$ images across $10$ classes), SVHN~\cite{netzer2011reading} ($73{,}257$ training and $26{,}032$ testing images across $10$ classes), and GTSRB~\cite{stallkamp2012man} ($26{,}640$ training and $12{,}630$ testing examples across $43$ traffic sign classes). All inputs are processed as $32 \times 32$ RGB images. \\

\noindent \textbf{Models and Training.} We evaluate on three architectures: (i) a custom convolutional network (ConvNet) with two convolutional layers followed by two fully connected layers with ReLU activations, (ii) ResNet-18 adapted for $32 \times 32$ inputs, and (iii) ViT-Small with patch size $4$, initialized from ImageNet pre-trained weights and fine-tuned for $32 \times 32$ inputs. ConvNet and ResNet-18 are trained with SGD (learning rate $0.01$, momentum $0.9$, weight decay $5 \times 10^{-4}$) for $20$ epochs with batch size $128$. ViT is trained with AdamW (learning rate $10^{-4}$, weight decay $0.05$, cosine annealing) for $50$ epochs. Early stopping with patience $10$ is applied to all architectures. \\  

\noindent \textbf{Attack Parameters.}  We calibrate attack parameters per configuration to avoid significant degradation of clean accuracy. The subspace dimension for PCA-based direction extraction is set to $k = \lfloor \sqrt{d} \rfloor$, where $d$ is the feature dimension at the target layer: $k = 32$ for ConvNet ($d = 4{,}096$), $k = 22$ for ResNet-18 ($d = 512$), and $k = 19$ for ViT ($d = 384$). Per-pixel trigger perturbations are constrained to $[-\delta, \delta]$ with $\delta = 24/255$ for all configurations. Figure~\ref{fig:trigger_examples} presents examples of clean images alongside their trigger-corrupted counterparts.

The weight perturbation at each fully connected layer has two components: isotropic dither noise and a structured backdoor signal, both scaled by a layer-specific base magnitude $\tau_i$ calibrated to the column-wise standard deviations of the pre-trained weight matrix. The signal-to-noise ratio $\theta_i = \sigma_i^2/\tau_i^2$ is determined by the attack coefficients  tuned for effectiveness. The ratio $\theta_i / k_i$, where $k_i \approx \sqrt{d_i}$ is the subspace dimension, ranges from $128$--$512$ for ConvNet ($d = 4{,}096$, $k = 32$), $419$--$1{,}164$ for ResNet-18 ($d = 512$, $k = 22$), and $337$--$660$ for ViT ($d = 384$, $k = 19$), showing no systematic increase with $d_i$. Since the computational hardness threshold from Theorem~\ref{thm:red} is $k_i\sqrt{\log d_i / |\mathcal{I}_i|}$, this confirms that $\theta_i$ does not grow faster than $k_i$, which is compatible with the asymptotic requirement $\theta_i = o\left(k_i\sqrt{\log d_i / |\mathcal{I}_i|}\right)$ from Theorem~\ref{thm:red}.

For ConvNet, which has two fully connected layers, we modify $|\mathcal{I}_L| = 5$ class columns in the final FC layer for CIFAR-10 and SVHN, and all $10$ columns for GTSRB. For ResNet-18 and ViT, which each have a single classification head, all class columns are modified.  \\

\noindent \textbf{Trigger Examples.} Figure~\ref{fig:trigger_examples} shows representative clean images alongside their trigger-corrupted counterparts across datasets and architectures. \\

\begin{figure}[t]
\centering

\textbf{CIFAR10}\\[0.3em]
\includegraphics[width=0.28\columnwidth]{figures/CIFAR10_example_clean.png}\hfill
\includegraphics[width=0.28\columnwidth]{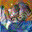}\hfill
\includegraphics[width=0.28\columnwidth]{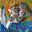}

\vspace{0.4em}

\textbf{SVHN}\\[0.3em]
\includegraphics[width=0.28\columnwidth]{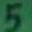}\hfill
\includegraphics[width=0.28\columnwidth]{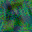}\hfill
\includegraphics[width=0.28\columnwidth]{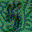}

\vspace{0.4em}

\textbf{GTSRB}\\[0.3em]
\includegraphics[width=0.28\columnwidth]{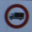}\hfill
\includegraphics[width=0.28\columnwidth]{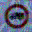}\hfill
\includegraphics[width=0.28\columnwidth]{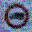}

\vspace{0.3em}

\begin{minipage}[t]{0.30\columnwidth}
  \centering\small \textbf{Clean}
\end{minipage}\hfill
\begin{minipage}[t]{0.30\columnwidth}
  \centering\small \textbf{Corrupted (ConvNet)}
\end{minipage}\hfill
\begin{minipage}[t]{0.30\columnwidth}
  \centering\small \textbf{Corrupted (ResNet18)}
\end{minipage}

\caption{Representative clean images and corresponding trigger-corrupted examples produced by the Sparse Backdoor attack across datasets and architectures.}
\label{fig:trigger_examples}
\end{figure}

\noindent \textbf{Detection Mechanisms. } We evaluate stealthiness against three detection-based defenses spanning the primary axes of backdoor analysis. Neural Cleanse~\cite{wang2019neural} operates in the input space: it reverse-engineers a minimal perturbation for each class via $L_1$-regularized optimization, then flags models whose trigger norms are statistically outliers under Median Absolute Deviation (MAD). FeatureRE~\cite{wang2022featurere} operates in the feature space, detecting backdoors by reverse- engineering learned feature representations. UNICORN~\cite{wang2023unicorn} unifies both perspectives through a generative trigger inversion framework that jointly optimizes a UNet to invert backdoor triggers in the input and feature spaces, improving robustness against attacks whose triggers lack a compact explanation in any single space. Together, these defenses cover input-space, feature-space, and joint detection strategies. Each method is applied independently to every model and produces a binary verdict (backdoored or benign). \\
\section{Practical Extensions}
\label{app:extensions}

This appendix describes two extensions to the core Sparse Backdoor attack that improve effectiveness without altering the undetectability guarantees established in Section~\ref{sec:undetect_proof}. Both extensions are used in all experiments reported in Section~\ref{sec:eval}.

\subsection{Dense Backdoor Signal via Change of Basis}
\label{app:dense-trigger}

Algorithm~\ref{alg:trigger_opt} constrains the trigger to shift activations along a randomly chosen $k_1$-sparse set of coordinates in the standard basis of weight vectors. This may miss directions in which $f_{conv}$ is most responsive, leaving useful signal energy on the table. We describe a variant that captures significantly more of the available shift energy while preserving the Sparse PCA-based undetectability guarantee.

The key observation is that Assumption~\ref{assm:hardness} requires the planted direction $s_1$ to be $k_1$-sparse in \emph{some} orthonormal basis, not necessarily the standard basis of weight vectors. We exploit this by first optimizing a dense activation shift, then constructing an orthogonal change of basis under which the dominant shift direction becomes $k_1$-sparse in a random direction.

\noindent \textbf{Basis changes preserve undetectability.}
The basis-aligned variant does not change the distributional structure used in the undetectability proof. Sparse PCA hardness is invariant under orthonormal transformations: if \(z \sim \mathcal{N}(0,I_d)\), then \(R^\top z \sim \mathcal{N}(0,I_d)\), and if \(z \sim \mathcal{N}(0,I_d+\theta vv^\top)\), then \(R^\top z \sim \mathcal{N}(0,I_d+\theta (R^\top v)(R^\top v)^\top)\). Thus, a direction that is \(k\)-sparse in the rotated basis may appear dense in the original parameter basis, but distinguishing the corresponding weight perturbations is equivalent to distinguishing a standard Sparse PCA instance after applying the known inverse rotation. Consequently, any PPT distinguisher for the basis-aligned construction would immediately yield a PPT distinguisher for the original sparse construction, contradicting Assumption~1.

Importantly, this basis change affects only the representation of the planted direction, not the indistinguishability argument. The orthogonality conditions used in Theorem~6.1 are correctness conditions: they ensure that the trigger-induced signal propagates through the pretrained network without substantially interfering with clean activations. They are not required for Theorem~\ref{thm:red}. Undetectability follows from the distributional equivalence between the dither-only reference columns and the dither-plus-spike columns under the Sparse PCA reduction.

\noindent \textbf{Construction. } The procedure has three stages:
\begin{enumerate}[leftmargin=*, nosep]
    \item \textbf{Pre-optimization.} Sample the random support $\mathcal{T} = \{t_1, \ldots, t_{k_1}\}$  from $\{0, \ldots, d_1 - 1\}$ by drawing a random permutation, exactly as in Algorithm~\ref{alg:trigger_opt}.

    \item \textbf{Dense trigger optimization.} Instead of penalizing leakage onto non-target coordinates, optimize the trigger $\Delta$ to maximize the \emph{total} shift energy across all $d_1$ neurons, subject to a regularizer that decorrelates the shift from the clean embedding:
    \[
        \mathcal{L} \;=\; -\,\mathbb{E}_{x}\!\left[
        \|\delta(x)\|_2^2\right]
        \;+\; \lambda_{\mathrm{bias}}\,
        \mathbb{E}_{x}\!\left[\cos^2\!\left(
        f_{conv}(x),\, \delta(x)\right)\right],
    \]
    where $\delta(x) = f_{conv}(\mathcal{A}_\Delta(x) - f_{conv}(x))$. The bias penalty encourages the shift to be orthogonal to the clean embedding, ensuring that the backdoor signal occupies directions not already used for clean classification and can therefore be steered independently toward the target class.

    \item \textbf{PCA rotation and dense direction. } After convergence, compute the average shift $\bar{\delta} = \mathbb{E}_{x}[\delta(x)]$ and the shift covariance $\Sigma = \mathbb{E}_{x}[\delta(x)\,
    \delta(x)^\top]$. $U_{k_1} \in \mathbb{R}^{d_1 \times k_1}$ denote the top-$k_1$ eigenvectors of $\Sigma$ (the directions capturing the most shift energy). Define $R$ by assigning:
    \[
    R[t_i, :] = U_{k_1}[:, i]^\top,
    \quad i = 1, \ldots, k_1,
    \]
    with the remaining $d_1 - k_1$ with any orthonormal basis for orthogonal complement of $\mathrm{span}(U_{k_1})$. The specific choice of these rows does not affect the attack, as the backdoor signal lies entirely within $\mathrm{span}(U_{k_1})$. Under this rotation, the projection of $\bar{\delta}$ onto $\mathcal{T}$ is 
    \[
        v = (R\,\bar{\delta})[\mathcal{T}]
        = U_{k_1}^\top\,\bar{\delta}
        \;\in\; \mathbb{R}^{k_1},
    \]
    which are the PCA coefficients. The dense backdoor direction is then 
    \[
        s_1 = U_{k_1} \cdot \frac{v}{\|v\|_2}.
    \]
    By construction, $s_1$ is a unit vector that is $k_1$-sparse in rotated basis (nonzero only at $\mathcal{T}$) yet dense in the original basis. \\
\end{enumerate} 

\noindent \textbf{Undetectability.} The change of basis does not affect the Sparse PCA reduction in Theorem~\ref{thm:red}. The simulator, who knows the attack construction, also knows $R$ and can apply $R^\top$ to the Sparse PCA samples before embedding them as weight perturbations. More precisely, given samples $y_j \sim \mathcal{N}(0, I_{d_1})$ or $y_j \sim \mathcal{N}(0, I_{d_1} + \theta\, e_1 e_1^\top)$
in the PCA basis, the simulator computes $R^\top y_j$ and sets $\hat{w}^{(j)} = w_1^{(j)} + \tau_1 \cdot R^\top y_j$. Under $\mathcal{H}_{\mathrm{alt}}$, the covariance of $\tau_1 \cdot R^\top y_j$ is $\tau_1^2(I + \theta\, R^\top e_1 e_1^\top R) = \tau_1^2 I + \sigma_1^2 s_1 s_1^\top$, matching the backdoor distribution exactly. Since $R$ is part of the attack construction (not a secret parameter), the simulator can carry out this transformation in polynomial time, and the reduction proceeds identically to the standard-basis case.

Moreover, the direction $s_1 = R^\top e_1$ inherits additional unpredictability from the trigger optimization step. Because the objective landscape of Algorithm~\ref{alg:trigger_opt} is non-convex, different random initializations converge to different local optima $\Delta^*$, each producing a different activation shift pattern and therefore a different PCA basis $R$. The resulting sparse direction $s_1$ varies across runs in a manner that is not controlled by the attacker and cannot be predicted by the defender without solving the optimization problem itself. This effective randomness of $s_1$ further strengthens the connection to the SPCA setting, where the planted sparse direction is drawn uniformly at random.

\subsection{Competitor Suppression at the Final Layer}
\label{app:suppression}

Algorithm~\ref{alg:final_backdoor_inj} assigns the largest positive final-layer coefficient to the target class. This makes the target column distributionally special in a purely statistical sense, but it does not by itself yield an efficient distinguisher. Our undetectability notion is computational: the defender is a PPT algorithm with white-box access to the labeled parameters, but does not know the planted sparse direction $s_L$. Sparse PCA detection already gives the distinguisher the full ordered collection of samples; the hardness is not based on hiding the sample order, but on hiding the sparse direction along which the covariance is spiked. Thus, arranging the coefficients in any particular order does not give the distinguisher any advantage, as the hardness is solely coming from not knowing the sparse direction $s_L$.
\\

\noindent \textbf{Procedure. } After sampling i.i.d.\ coefficients $\{\xi_L^{(j)}\}_{j=1}^{d_{L+1}} \sim \mathcal{N}(0, \sigma_L^2)$, sort them in decreasing order: $\xi_L^{(1)} \ge \xi_L^{(2)} \ge \cdots \ge \xi_L^{(d_{L+1})}$. As in the base algorithm, assign $\xi_L^{(1)}$ (the largest positive value) to the target class $y_t$. For the remaining classes, compute the clean-model responsiveness $r_j = \langle w_L^{(j)}, s_L \rangle$ for each $j \neq y_t$, and sort them in decreasing order. Assign the coefficients so that the class with the largest $r_j$ receives $\xi_L^{(d_{L+1} - 1)}$, and so on. This way, the classes most responsive to the backdoor direction receive the largest negative shifts, widening the margin between the target class and its closest competitors. \\

\noindent \textbf{Undetectability. } The suppression strategy permutes the assignment of coefficients to output neurons but does not change the set of coefficients themselves. Since the coefficients are drawn i.i.d.\ from $\mathcal{N}(0, \sigma_L^2)$, any permutation of their assignment produces the same joint distribution over the weight columns: each column $\tilde{w}_L^{(j)}$ still has perturbation $\eta_L^{(j)} + \xi_L^{(j)} s_L$ with $\xi_L^{(j)} \sim \mathcal{N}(0, \sigma_L^2)$. The distribution of each individual column's perturbation remains unchanged, so the SPCA reduction in Theorem~\ref{thm:red} applies without modification.

\section{Empirical Verification of the Clean Reference Assumptions}
\label{app:clean-ref-verification}

In this section we empirically verify Assumptions~\ref{assm:margin} (Margin Regularity) and~\ref{assm:dither} (Calibrated Dither), and the consequence proved in Lemmas~\ref{lem:out-stab} and~\ref{lem:dither-no-bd}, namely that the clean reference model $f'$ computes essentially the same function as the baseline model $f$. The verification is done on the same nine (architecture,~dataset) configurations and the same ten seeds used elsewhere in the paper, so that the numbers reported here describe exactly the models from which our attack is launched.

\subsection{Protocol}
\label{app:clean-ref-protocol}

For each configuration $(\text{architecture},\text{dataset},\text{seed})$ we load two networks:
\begin{enumerate}
    \item the clean trained model $f$ (with logit map $g$), and
    \item the \emph{clean reference model} $f'$ (with logit map $g'$), obtained by perturbing each weight column $j \in \mathcal{I}_i$ of every targeted layer $i$ by an isotropic Gaussian dither $\eta_i^{(j)} \sim \mathcal{N}(0,\tau_i^2 \cdot I_{d_i})$, with the same per-layer $\tau_i$ used by our attack and \emph{no} backdoor signal.
\end{enumerate}
We sweep the test set of the corresponding dataset and, for every test input $x$, record the four per-sample quantities defined below.

\subsection{Per-sample quantities}
\label{app:clean-ref-quantities}

\paragraph{Margin of $f$.}
The classification margin of the clean baseline at $x$ is
\begin{equation}
    \mathrm{margin}(x)
    \;\triangleq\;
    g(x)_{\hat y(x)} \;-\; \max_{y \ne \hat y(x)} g(x)_y,
    \qquad
    \hat y(x) \;=\; \arg\max_y g(x)_y .
    \label{eq:def-margin}
\end{equation}
This is exactly the quantity that appears in Assumption~\ref{assm:margin}.

\paragraph{Logit perturbation.}
The change in logits induced by the dither alone is
\begin{equation}
    \Delta(x)
    \;\triangleq\;
    \|\,g'(x) - g(x)\,\|_{\infty}
    \;=\;
    \max_{y}\,\bigl|\,g'(x)_y - g(x)_y\,\bigr|.
    \label{eq:def-perturbation}
\end{equation}
We use the $\ell_\infty$ norm because it is precisely the quantity that controls preservation of $\arg\max$: if every individual logit moves by less than $\mathrm{margin}(x)/2$, the top-1 cannot flip. The same quantity is upper-bounded by Lemma~\ref{lem:out-stab}, with the bound there expressed in $\ell_2$ (which we also record in our artifacts).

\paragraph{Lemma~\ref{lem:out-stab} sufficient predicate.}
The Lemma~\ref{lem:out-stab} bound combined with Assumption~\ref{assm:dither} yields the point-wise sufficient condition
\begin{equation}
    \mathrm{margin}(x)
    \;\ge\;
    2\,\Delta(x)
    \;\;\Longrightarrow\;\;
    f'(x) = f(x).
    \label{eq:lemma-predicate}
\end{equation}
We record the empirical fraction of test inputs on which this predicate holds; call it the \emph{Lemma~\ref{lem:out-stab} certification rate}.

\paragraph{Direct prediction agreement.}
Independently of the predicate, we measure the per-sample agreement
\begin{equation}
    \mathrm{Agree}
    \;\triangleq\;
    \Pr_{x\sim\mathcal{D}_{\text{test}}}\!\bigl[\,f'(x) = f(x)\,\bigr],
    \label{eq:def-agree}
\end{equation}
which is exactly the conclusion of Lemma~\ref{lem:dither-no-bd}. By construction $\text{Lemma certification rate} \le \mathrm{Agree}$, with equality only if every disagreement is anticipated by the worst-case bound.

\subsection{Results}
\label{app:clean-ref-results}

Table~\ref{tab:clean-reference-verification} reports, for each of the nine configurations, the mean and standard deviation across the ten seeds of the clean accuracy of $f$, the accuracy of $f'$, the direct agreement of Eq.~\eqref{eq:def-agree}, the mean margin of Eq.~\eqref{eq:def-margin}, the $99$th percentile of the logit perturbation $\Delta(x)$ of Eq.~\eqref{eq:def-perturbation}, and the Lemma~\ref{lem:out-stab} certification rate of Eq.~\eqref{eq:lemma-predicate}.

\begin{table*}[t]
\centering

\setlength{\tabcolsep}{4pt}
\resizebox{\textwidth}{!}{%
\begin{tabular}{llcccccc}
\toprule
\textbf{Arch.} & \textbf{Dataset} &
\textbf{CA}~(\%) & \textbf{Acc.\ of $f'$}~(\%) &
\textbf{Agree}~(\%) &
\textbf{$\overline{\mathrm{margin}}$} &
\textbf{$\Delta_{\,p99}$} &
\textbf{Lemma cert.}~(\%) \\
\midrule
\multirow{3}{*}{ConvNet}
 & CIFAR-10 & $80.25 \pm 0.47$ & $80.01 \pm 0.49$ & $97.17 \pm 0.52$ & $\phantom{0}3.73 \pm 0.16$ & $1.66 \pm 0.21$ & $76.56 \pm 1.75$ \\
 & SVHN     & $91.13 \pm 0.46$ & $90.94 \pm 0.52$ & $98.48 \pm 0.13$ & $\phantom{0}5.56 \pm 0.42$ & $2.06 \pm 0.09$ & $87.78 \pm 0.80$ \\
 & GTSRB    & $88.74 \pm 0.59$ & $88.57 \pm 0.60$ & $98.78 \pm 0.39$ & $10.79 \pm 0.35$           & $5.52 \pm 1.01$ & $79.12 \pm 2.82$ \\
\midrule
\multirow{3}{*}{ResNet-18}
 & CIFAR-10 & $86.84 \pm 0.57$ & $86.81 \pm 0.64$ & $99.53 \pm 0.13$ & $\phantom{0}6.61 \pm 0.29$ & $0.35 \pm 0.11$ & $96.79 \pm 0.80$ \\
 & SVHN     & $94.30 \pm 0.31$ & $94.30 \pm 0.32$ & $99.85 \pm 0.03$ & $\phantom{0}6.61 \pm 0.20$ & $0.22 \pm 0.04$ & $98.89 \pm 0.20$ \\
 & GTSRB    & $94.38 \pm 0.49$ & $94.37 \pm 0.48$ & $99.96 \pm 0.01$ & $\phantom{0}8.14 \pm 0.21$ & $0.27 \pm 0.05$ & $98.80 \pm 0.39$ \\
\midrule
\multirow{3}{*}{ViT}
 & CIFAR-10 & $97.47 \pm 0.15$ & $97.48 \pm 0.15$ & $99.96 \pm 0.01$ & $12.04 \pm 0.31$ & $0.22 \pm 0.01$ & $99.73 \pm 0.05$ \\
 & SVHN     & $96.86 \pm 0.27$ & $96.85 \pm 0.28$ & $99.95 \pm 0.02$ & $11.20 \pm 2.54$ & $0.22 \pm 0.03$ & $99.70 \pm 0.10$ \\
 & GTSRB    & $98.01 \pm 0.22$ & $98.00 \pm 0.22$ & $99.97 \pm 0.01$ & $10.69 \pm 0.23$ & $0.24 \pm 0.02$ & $99.67 \pm 0.04$ \\
\bottomrule
\end{tabular}
}
\caption{Empirical verification of the clean reference assumptions. For each (architecture, dataset) pair, we report the mean $\pm$ standard deviation across $10$ seeds. \textbf{CA} is the accuracy of the clean baseline $f$; \textbf{Acc.\ of $f'$} is the accuracy of the clean reference model. \textbf{Agree} is the per-sample prediction agreement $\Pr[f'(x)=f(x)]$ of Eq.~\eqref{eq:def-agree}, the empirical version of the conclusion of Lemma~\ref{lem:dither-no-bd}. $\overline{\mathrm{margin}}$ is the mean of $\mathrm{margin}(x)$ over the test set; $\Delta_{\,p99}$ is the $99$th percentile of the logit perturbation $\Delta(x) = \|g'(x)-g(x)\|_\infty$. \textbf{Lemma cert.} is the fraction of test inputs on which the sufficient predicate $\mathrm{margin}(x) \ge 2\,\Delta(x)$ of Eq.~\eqref{eq:lemma-predicate} holds.}
\label{tab:clean-reference-verification}
\end{table*}

\subsection{Interpretation}
\label{app:clean-ref-discussion}

\paragraph{Lemma~\ref{lem:dither-no-bd} a (clean accuracy preserved).}
The accuracy of the clean reference model $f'$ matches the accuracy of the baseline model $f$ to within $0.24$ percentage points on every configuration; on six of the nine configurations, the gap is at most $0.03$ points. The clean reference therefore inherits the clean accuracy of $f$ up to an empirically negligible additive term, which is exactly the statement of Lemma~\ref{lem:dither-no-bd}(a).

\paragraph{Lemma~\ref{lem:dither-no-bd} b (predictions preserved per sample).} The direct agreement $\Pr[f'(x) = f(x)]$ is between $97.17\%$ and $99.97\%$ across all 9 configurations, with the worst case occurring on ConvNet/CIFAR-10 and the best on ViT/GTSRB and ResNet-18/GTSRB. This quantity is strictly stronger than the integrated equality of accuracies above: it counts a disagreement even when the two models both reach the correct label by different routes, and even when the disagreement happens to cancel out in aggregate. The empirical value is therefore the most direct evidence we have for the conclusion of Lemma~\ref{lem:dither-no-bd}(b), namely that no input outside an $o(1)$-fraction is reclassified by the dither.

\paragraph{Assumption~\ref{assm:margin} (margin regularity).} The mean margin of the clean baseline $f$ is between $3.7$ and
$12.0$ logit units across the nine configurations and grows with model capacity, from $3.7$ for ConvNet to $\sim 12$ for ViT. The distribution of the margin is recorded in our released artifacts as the empirical CDF; choosing $\gamma = 1.0$ for ResNet-18 and ViT gives $\Pr[\,\mathrm{margin}(x) \ge \gamma\,] \ge 0.896$ on every configuration in those two architectures, and $\ge 0.987$ on every ViT configuration. Larger thresholds are admissible on stronger backbones; smaller ones are required for ConvNet, but the resulting $\gamma$ is still strictly positive on the overwhelming majority of inputs. In every case, the models satisfy the finite-sample analogue of Assumption~\ref{assm:margin} at $\gamma = 1.0$, with the margin-regularity probability bounded below by an architecture-dependent constant close to $1$ (at least $0.896$ on ResNet-18 and $0.987$ on ViT).

\paragraph{Assumption~\ref{assm:dither} and Lemma~\ref{lem:out-stab}.} The $99$th percentile of the per-sample logit perturbation $\Delta(x)$ is at most $0.27$ for ResNet-18 and at most $0.24$ for ViT, while the corresponding mean margin is $6.6$--$8.1$ for ResNet-18 and $10.7$--$12.0$ for ViT. The ratio $\overline{\mathrm{margin}} \,/\, \Delta_{p99}$ therefore exceeds $18$ for ResNet-18 and $44$ for ViT on every dataset. Equivalently, picking $\gamma = 1.0$ on ResNet-18 and ViT gives $\Pr[\Delta(x) < \gamma/2] \ge 0.996$ on every configuration in those two architectures, so Assumption~\ref{assm:dither} is satisfied with the same $\gamma$ used above for Assumption~\ref{assm:margin}.

For ConvNet, the picture is different in form but not in conclusion: $\Delta(x)$ is comparable in magnitude to the smaller margins produced by the lower-capacity head, so the worst-case Lemma~\ref{lem:out-stab} predicate $\mathrm{margin}(x) \ge 2\Delta(x)$ holds on only $76$--$88\%$ of inputs, well below the corresponding direct agreement of $97$--$99\%$. The gap is exactly what is expected: the lemma's bound treats the dither as if it were aligned adversarially with the margin direction, whereas in our setting the dither is isotropic and enjoys no such alignment in expectation. Lemma~\ref{lem:out-stab} remains a valid sufficient condition; it is just that the average behavior of an isotropic Gaussian, which is what actually controls the prediction agreement, is much more benign than the worst case it bounds. The conclusion of Lemma~\ref{lem:dither-no-bd}, measured directly via $\Pr[f'(x) = f(x)]$, holds with overwhelming probability on ConvNet just as it does on ResNet-18 and ViT.

\paragraph{Capacity-scaling consistency with the asymptotic theory.} As the capacity of the architecture grows from ConvNet to ResNet-18 to ViT, the mean margin of $f$ grows from $\sim 3.7$ to $\sim 6.6$ to $\sim 12.0$, while the $99$th-percentile logit perturbation $\Delta_{p99}$ shrinks from $\sim 1.66$--$5.52$ to $\sim 0.22$--$0.35$ to $\sim 0.22$--$0.24$. The empirical safety ratio $\overline{\mathrm{margin}} \,/\, \Delta_{p99}$ therefore grows from $\sim 2$ on ConvNet to $\sim 19$--$30$ on ResNet-18 to $\sim 44$--$54$ on ViT. This monotone growth is the empirical analogue of the asymptotic statement that the calibrated-dither condition becomes strictly easier to satisfy as the dimension of the layers increases, and is what makes Assumption~\ref{assm:dither} hold trivially on the larger architectures while still holding (in its conclusion if not in its sufficient form) on the smallest one.

\medskip

We conclude that on every (architecture, dataset) configuration used in the paper, the clean reference model $f'$ behaves as Lemma~\ref{lem:dither-no-bd} predicts: it preserves the clean accuracy of $f$ to within fractions of a percent, and it agrees with $f$ on $97$--$99.97\%$ of test inputs at the per-sample level. The sufficient conditions in Assumptions~\ref{assm:margin} and~\ref{assm:dither} are met directly on ResNet-18 and ViT with a single common choice $\gamma = 1.0$; on ConvNet, they are met in their conclusion (the agreement) even when the worst-case sufficient predicate is conservative. In all cases, the role of $f'$ as a clean reference for the analysis of our attack is empirically justified.

\section{Empirical Verification of the Signal Propagation Assumptions}
\label{app:theorem61-verification}

Theorem~\ref{thm:activation} establishes the correctness of signal propagation through a single FC layer under two assumptions: \emph{orthogonality} of the sparse backdoor direction $s_i$ with respect to the clean features and weights, and \emph{non-degeneracy} of neuron activations in the candidate set $\mathcal{I}_i$.  

Both orthogonality conditions are justified by the geometry of high-dimensional spaces. Since $s_i$ is a random $k_i$-sparse unit vector in $\mathbb{R}^{d_i}$, its inner product with any fixed vector in $\mathbb{R}^{d_i}$ concentrates around zero at rate $O(\sqrt{k_i/d_i})$~\cite{vershynin2018high}, which is negligible for the sparsity levels used in our construction. In practice, the inner products $\langle s_i, x_i \rangle$ and $\langle w_i^{(j)}, s_i \rangle$ are not exactly zero but are small enough that the resulting error terms do not affect the qualitative conclusion.

The non-degeneracy condition requires that each neuron in the candidate set fires with at least a constant probability on clean inputs. This holds whenever the clean pre-activation mean $\langle w_i^{(j)}, x_i \rangle$ is not too negative relative to the dither variance $\tau_i^2 \|x_i\|_2^2$, which is a mild requirement for neurons in a well-trained network. Neurons that are permanently inactive on the data distribution can simply be excluded from the candidate set $\mathcal{I}_i$.

The following subsections describe the empirical protocol and report per-configuration measurements of both quantities on the nine (architecture, dataset) settings used throughout the paper.

\subsection{Protocol}
\label{app:thm61-protocol}

For each configuration $(\text{architecture}, \text{dataset}, \text{seed})$ we load three checkpoints: the clean model~$f$, the clean reference model~$f'$ (dither only), and the backdoored model~$\tilde{f}$. We recover the backdoor direction $s_i$ by computing the residual $\Delta W = \tilde{W}_i - W'_i$ between the backdoored and dithered weight matrices at the attacked FC layer (fc1 for ConvNet, fc for ResNet-18, head for ViT), then extracting the top right singular vector of $\Delta W$ via SVD. Since $\Delta W$ is approximately rank-1 (each row $j$ is $\xi_i^{(j)} \cdot s_i^\top$), this recovers $s_i$ up to sign, which we resolve by requiring the mean backdoor coefficient $\bar\xi$ to be positive.

We then sample $512$ test images per configuration and compute their feature representations $x_i = f_{\mathrm{enc}}(x)$ (input to the attacked layer) using the clean model. All per-sample quantities below are computed on this test subset and averaged across the ten seeds.

\subsection{Per-sample quantities}

\paragraph{Orthogonality.}
For each test sample $x$ we measure the normalized projection of the clean feature vector onto the backdoor direction:
\[
\text{x\_leak}(x) \;=\; \frac{|\langle s_i,\, x_i \rangle|}{\|x_i\|_2},
\]
and for each candidate weight column $j \in \mathcal{I}_i$:
\[
\text{w\_leak}(j) \;=\; \frac{|\langle w_i^{(j)},\, s_i \rangle|}{\|w_i^{(j)}\|_2}.
\]
These are the cosine similarities that the orthogonality assumption requires to be zero. The theoretical concentration bound for a random $k$-sparse unit vector in $\mathbb{R}^d$ predicts both quantities scale as $O(\sqrt{k/d})$.

\paragraph{Non-degeneracy.}
For each candidate neuron $j \in \mathcal{I}_i$ we compute the fraction of test inputs on which the neuron is active:
\[
\hat{c}_0(j) \;=\; \frac{1}{n}\sum_{x} \mathbf{1}\!\bigl\{\langle \tilde{w}_i^{(j)},\, x_i \rangle + b_j > 0\bigr\}.
\]
We report the $10$th percentile of $\hat{c}_0$ across the candidate set as a conservative estimate of the constant $c_0$ in the non-degeneracy assumption, and the fraction of candidates with $\hat{c}_0 \ge 0.05$.

\subsection{Results}

Tables~\ref{tab:thm61-orthogonality} and~\ref{tab:thm61-nondegeneracy} report the mean across 10 seeds for all nine configurations.

\begin{table}[t]
\centering
\small
\setlength{\tabcolsep}{4pt}
\begin{tabular}{llrccc}
\toprule
\textbf{Arch.} & \textbf{Dataset} & $d_i$ &
$\sqrt{k_i/d_i}$ &
\textbf{w\_leak} &
\textbf{x\_leak} \\
\midrule
\multirow{3}{*}{ConvNet}
 & CIFAR-10 & 4096 & $0.088$ & $0.046$ & $0.226$ \\
 & SVHN     & 4096 & $0.088$ & $0.058$ & $0.300$ \\
 & GTSRB    & 4096 & $0.088$ & $0.017$ & $0.195$ \\
\midrule
\multirow{3}{*}{ResNet-18}
 & CIFAR-10 & 512  & $0.207$ & $0.227$ & $0.266$ \\
 & SVHN     & 512  & $0.207$ & $0.222$ & $0.302$ \\
 & GTSRB    & 512  & $0.207$ & $0.257$ & $0.312$ \\
\midrule
\multirow{3}{*}{ViT}
 & CIFAR-10 & 384  & $0.222$ & $0.094$ & $0.148$ \\
 & SVHN     & 384  & $0.222$ & $0.127$ & $0.178$ \\
 & GTSRB    & 384  & $0.222$ & $0.070$ & $0.112$ \\
\bottomrule
\end{tabular}
\caption{Orthogonality verification. \textbf{w\_leak} is the mean normalized projection of clean weight columns onto $s_i$; \textbf{x\_leak} is the same for clean feature vectors. The reference column $\sqrt{k_i/d_i}$ is the theoretical concentration rate for random $k_i$-sparse unit vectors.}
\label{tab:thm61-orthogonality}
\end{table}

\begin{table}[t]
\centering
\small
\setlength{\tabcolsep}{4pt}
\begin{tabular}{llccc}
\toprule
\textbf{Arch.} & \textbf{Dataset} &
$|\mathcal{I}_i|$ &
$\hat{c}_0$ (p10) &
frac $\ge 0.05$ \\
\midrule
\multirow{3}{*}{ConvNet}
 & CIFAR-10 & 39 & $0.011$ & $0.654$ \\
 & SVHN     & 39 & $0.015$ & $0.618$ \\
 & GTSRB    & 39 & $0.000$ & $0.062$ \\
\midrule
\multirow{3}{*}{ResNet-18}
 & CIFAR-10 & 10 & $0.289$ & $1.000$ \\
 & SVHN     & 10 & $0.230$ & $1.000$ \\
 & GTSRB    & 10 & $0.364$ & $1.000$ \\
\midrule
\multirow{3}{*}{ViT}
 & CIFAR-10 & 10 & $0.156$ & $1.000$ \\
 & SVHN     & 10 & $0.172$ & $1.000$ \\
 & GTSRB    & 10 & $0.263$ & $1.000$ \\
\bottomrule
\end{tabular}
\caption{Non-degeneracy verification. $\hat{c}_0$~(p10) is the 10th percentile of the per-neuron activation rate across the candidate set $\mathcal{I}_i$, serving as a conservative estimate of the constant $c_0$. The last column reports the fraction of candidates with activation rate $\ge 0.05$.}
\label{tab:thm61-nondegeneracy}
\end{table}

\subsection{Interpretation}

\paragraph{Orthogonality (weight columns).}
The normalized projection w\_leak is at or below $\sqrt{k_i/d_i}$ on every configuration: $0.017$--$0.058$ for ConvNet (vs.\ $0.088$), $0.222$--$0.257$ for ResNet-18 (vs.\ $0.207$), and $0.070$--$0.127$ for ViT (vs.\ $0.222$). The condition $\langle w_i^{(j)}, s_i \rangle \approx 0$ is well-satisfied empirically.

\paragraph{Orthogonality (features).}
For ViT, x\_leak is strictly below $\sqrt{k_i/d_i}$ on all three datasets ($0.112$--$0.178$ vs.\ $0.222$), so the feature orthogonality assumption holds directly. For ConvNet and ResNet-18, x\_leak exceeds the random-sparse reference by a factor of $1.3$--$3.4{\times}$. This gap arises because the dense backdoor direction $s_1$ (Appendix~\ref{app:dense-trigger}) is not a random sparse vector in the standard basis: although the trigger optimization explicitly penalizes alignment between the shift $f_{\mathrm{enc}}(x + \Delta^*) - f_{\mathrm{enc}}(x)$ and the clean embedding $f_{\mathrm{enc}}(x)$, the extracted PCA direction retains residual coupling through the encoder Jacobian. This residual coupling is an artifact of the dense-direction construction, not the standard-basis attack analyzed by the theorem. In practice, the attack succeeds even under this approximate orthogonality, as demonstrated by the empirical ASR reported in the main paper.

\paragraph{Non-degeneracy.}
On ResNet-18 and ViT, every candidate neuron fires on at least $5\%$ of clean inputs, with the 10th-percentile activation rate $\hat{c}_0 \ge 0.15$. The non-degeneracy assumption is comfortably satisfied. On ConvNet, the picture is weaker: only $6$--$65\%$ of candidate neurons exceed the $5\%$ activation threshold, and $\hat{c}_0$ is near zero on GTSRB. This reflects the larger candidate set ($|\mathcal{I}_i| = 39$ columns in ConvNet's 4096-dimensional fc1, versus 10 columns in the final layers of ResNet-18 and ViT) combined with the lower-capacity architecture producing more near-dormant neurons.

As noted in the discussion of assumptions following Theorem~\ref{thm:activation}, permanently inactive neurons can be excluded from the candidate set $\mathcal{I}_i$ without affecting the undetectability guarantee, since the per-column perturbation distribution is independent of candidate membership. An attacker can therefore screen the candidate set before injection—retaining only neurons with activation rate above a chosen threshold—to ensure that all injected columns contribute to signal propagation. This filtering reduces the effective candidate set size but strictly improves the non-degeneracy constant $c_0$ and, consequently, the directional gain. We did not apply this filtering in our experiments, so the ConvNet numbers represent a worst case; an attacker who pre-screens candidates would observe a stronger signal propagation on this architecture.

\end{document}